\documentclass[11pt]{article}
\usepackage{fullpage,amsfonts,amssymb,amsmath,amsthm,float,graphicx,bbm}
\usepackage[breaklinks]{hyperref}
\usepackage{algorithm}

\newtheorem{theorem}{Theorem}
\newtheorem{lemma}[theorem]{Lemma}
\newtheorem{lemmaPrime}{Lemma}
\newtheorem{corollary}{Corollary}
\theoremstyle{definition}
\newtheorem{definition}[theorem]{Definition}

\newcommand{\ket}[1]{\mathinner{\left \lvert #1 \right\rangle}}
\newcommand{\braket}[2]{\mathinner{ \left\langle #1 \right. \left \lvert #2 \right\rangle}}
\newcommand\lr[1]{\left( #1 \right)}
\newcommand\lra[1]{\left\langle  #1 \right\rangle}
\newcommand\lrv[1]{\left|  #1 \right|}
\newcommand\lrk[1]{\left[  #1 \right]}
\newcommand\lrb[1]{\left\lbrace #1 \right\rbrace}
\newcommand\nrm[1]{\left\Vert #1 \right\Vert}
\DeclareMathOperator{\diag}{diag}

\newcommand{\mb}[1]{\mathbf{#1}}

\newcommand{\vbl}{\vphantom{X^{X^X}}}
\newcommand{\tdots}{\,..\,}

\DeclareMathOperator*{\Span}{span}

\newcommand{\mbb}[1]{\mathbbm{#1}}
\newcommand{\mc}[1]{\mathcal{#1}}

\def\T{{\cal T}}
\def\A{{\cal A}}
\def\D{{\cal D}}
\def\G{{\cal G}}

\def\rket{\rangle}

\def\lbra{\langle}
\newcommand{\bra}[1]{\langle #1|}

\newcommand{\lmax}[1]{\lambda_{#1}}
\newcommand{\lRho}{\lmax{K}}

\newcommand{\noA}{A}
\newcommand{\noB}{B}
\newcommand{\noTot}{V}

\newcommand{\vertSet}{\mathsf V}
\newcommand{\edgeSet}{\mathsf E}
\newcommand{\vertSetA}{\mathsf V_{\mathsf A}}
\newcommand{\vertSetB}{\mathsf V_{\mathsf B}}

\newcommand{\spaceA}{\mc H_{\mathsf A}}
\newcommand{\spaceB}{\mc H_{\mathsf B}}
\newcommand{\RA}{R_{\mathsf A}}
\newcommand{\RB}{R_{\mathsf B}}
\newcommand{\matA}{M_{\mathsf A}}
\newcommand{\matB}{M_{\mathsf B}}
\newcommand{\mata}{M_{\mathsf a}}
\newcommand{\matb}{M_{\mathsf b}}

\newcommand{\addV}{v_{\noTot+1}}
\newcommand{\addE}{e_{T+1}}
\newcommand{\addVV}{v_{\noTot+2}}
\newcommand{\addEE}{e_{T+2}}

\newcommand{\complV}{ {\vertSet}'}
\newcommand{\complE}{ {\edgeSet}'}
\newcommand{\complG}{ {\mc G}'}

\newcommand{\complVV}{ {\vertSet}''}
\newcommand{\complEE}{ {\edgeSet}''}
\newcommand{\complGG}{ {\mc G}''}

\newcommand{\nhood}{\mathsf E}

\newcommand{\refA}{{\rm {ref}}_{\mc A}}
\newcommand{\refB}{{\rm {ref}}_{\mc B}}
\newcommand{\refX}[1]{{\rm {ref}}_{{#1}}}
\newcommand{\matD}{\mb D}

\newcommand{\tNAA}{\mb N}
\newcommand{\vecP}{\mb d}

\begin{document}

\setcounter{footnote}{2}
\renewcommand*{\thefootnote}{\fnsymbol{footnote}} 

\title{Quantum algorithm for tree size estimation, with applications to backtracking and 2-player games}
\date{}
\author{Andris Ambainis$^{\fnsymbol{footnote}}$
\and Martins Kokainis$^{\fnsymbol{footnote}}$
}
\maketitle

\begin{abstract}
\noindent
We study quantum algorithms on search trees of unknown structure, in a model where the tree can be discovered by local exploration. That is, we are given the root of the tree and access to a black box which, given a vertex $v$, outputs the children of $v$.
 
We construct a quantum algorithm which, given such access to a search tree of depth at most $n$,
estimates the size of the tree $T$ within a factor of $1\pm \delta$ in $\tilde{O}(\sqrt{nT})$ steps. More generally, the same algorithm can be used to estimate
size of directed acyclic graphs (DAGs) in a similar model.

We then show two applications of this result:
\begin{itemize}
\item
We show how to transform a classical backtracking search algorithm which examines $T$ nodes of a search tree into an $\tilde{O}(\sqrt{T}n^{3/2})$ time quantum algorithm, improving
over an earlier quantum backtracking algorithm of Montanaro \cite{Montanaro2015}. 
\item
We give a quantum algorithm for evaluating AND-OR formulas in a model where the formula can be discovered by local exploration (modeling position trees in 2-player games).
We show that, in this setting, formulas of size $T$ and depth $T^{o(1)}$ can be evaluated in quantum time $O(T^{1/2+o(1)})$.
Thus, the quantum speedup is essentially the same as in the case when the formula is known in advance.
\end{itemize}
\end{abstract}

\footnotetext[2]{Faculty of Computing, University of Latvia}
\renewcommand*{\thefootnote}{\arabic{footnote}} 
\setcounter{footnote}{0}

\setcounter{page}{0}
\thispagestyle{empty}
\pagebreak

\section{Introduction}

Many search algorithms involve exploring search trees of an unknown structure. For example, backtracking algorithms  
perform a depth-first search on a tree consisting of partial solutions to the computational task (for example, partial assignments for SAT in the well known DPLL algorithm \cite{DLL,DP}), until a full solution is found. 
Typically, different branches of the tree stop at different depths (e.g., when the corresponding partial assignment can no longer be extended) and the structure of the tree can be only determined by exploring it.  

Quantum algorithms provide a quadratic speedup for many search problems, from simple exhaustive search (Grover's algorithm \cite{Grover}) to computing AND-OR formulas \cite{A+,FGG,R} (which corresponds to determining the winner in a 2-player game, given a position tree). These algorithms, however, assume that the structure of the search space is known. Grover's algorithm assumes that the possible solutions in the search space
can be indexed by numbers $1, 2, \ldots, T$ so that, given $i$, one can efficiently (in constant or polylog time) find the $i^{\rm th}$ possible solution. In the case of backtracking trees, the unknown structure of the tree prevents us from setting up such an addressing scheme. 

In the case of AND-OR formulas, the quantum algorithms of \cite{A+,R} work for formula of any structure but
the coefficients in algorithm's transformations depend on the sizes of different subtrees of the formula tree.
Therefore, the algorithm can be only used if the whole AND-OR formula is known in advance (and only the 
values of the variables are unknown) which is not the case if the formula corresponds to a position tree in a game.

Despite the importance of such algorithms classically, there has been little work on quantum search on structures which can be only explored locally. 
The main algorithmic result of this type is a recent algorithm by Montanaro for quantum backtracking. Given a search tree of size $T$ and depth $n$, Montanaro's algorithm detects if the tree contains a marked vertex in $O(\sqrt{Tn})$ steps
(and finds a marked vertex in $O(\sqrt{T} n^{3/2})$ steps). 

In this paper, we show three new quantum algorithms for trees of an unknown structure, 
including an improvement to Montanaro's algorithm. We start with

{\bf Quantum tree size estimation}. We show that, given a tree $\T$ with a depth at most $n$, the size $T$ of the tree $\T$ can be estimated to a multiplicative factor of $1+\delta$, for an arbitrary constant $\delta>0$, by a quantum algorithm that uses $\tilde{O}(\sqrt{Tn})$ steps\footnote{In the statements of results in the abstract and the introduction, $\tilde{O}$ hides $\log T$ and $\log n$ factors and the dependence of the running time on the maximum degree $d$ of vertices in a tree $\T$ (or a DAG $\G$). More precise bounds
are given in Section \ref{sec:results}.}. 
More generally, our algorithm is also applicable to estimating size of directed acyclic graphs (DAGs) in a similar model.

We then apply the quantum tree size estimation algorithm to obtain two more results. 

{\bf Improved quantum algorithm for backtracking.}
Montanaro's algorithm has the following drawback. Since classical search algorithms are optimized to search the most promising branches first, a classical search algorithm may find a marked vertex after examining $T' \ll T$ nodes of the tree.  Since the running time of Montanaro's algorithm depends on $T$, the quantum speedup that it achieves can be much less than quadratic (or there might be no speedup at all).

We fix this problem by using our tree size estimation algorithm.
Namely, we construct a quantum algorithm that searches a backtracking tree in $\tilde{O}(\sqrt{T'} n^{3/2})$ steps where $T'$ is the number of nodes actually visited by the classical algorithm. 

{\bf AND-OR formulas of unknown structure.} We construct a quantum algorithm for computing AND-OR formulas in a model 
where the formula is accessible by local exploration, starting from the root (which is given).
More specifically, we assume query access to the following subroutines:
\begin{itemize}
\item
given a node $v$, we can obtain the type of the node (AND, OR or a leaf), 
\item
given a leaf $v$, we can obtain the value of the variable (true or false) at this leaf,
\item
given an AND/OR node, we can obtain pointers to the inputs of the AND/OR gate.
\end{itemize}
This models a position tree in a 2-player game (often mentioned as a motivating example for studying AND-OR trees) with OR gates corresponding to positions at which the 1$ ^\text{st} $ player makes a move  and 
AND gates corresponding to positions at which the 2$ ^\text{nd} $ player makes a move. 

We give an algorithm that evaluates AND-OR formulas of size $T$ and depth $T^{o(1)}$ in this model
with $O(T^{1/2+o(1)})$ queries. Thus, the quantum speedup is almost the same as in the case when
the formula is known in advance (and only values at the leaves need to be queried) \cite{A+,R}, 
as long as the depth of the tree is not too large.

\section{Preliminaries}

\subsection{Setting}
\label{sec:setting}

We consider a tree $\T$ of an unknown structure given to us in the following way:
\begin{itemize}
\item
We are given the root $r$ of $\T$.
\item
We are given a black box which, given a vertex $v$, returns the number of children $d(v)$ for this vertex.
\item
We are given a black box which, given a vertex $v$ and $i\in[d(v)]$, returns the $i^{\rm th}$ child of $v$.
\end{itemize}

Trees of unknown structure come up in several different settings.

{\bf Backtracking.} Let $\A$ be a backtracking algorithm that searches a solution space $\D$ in a depth-first fashion.
The space $\D$ consists of partial solutions where some of the relevant variables have been set. 
(For example, $\D$ can be the space of all partial assignments for a SAT formula.) 
Then, the corresponding tree $\T$ is defined as follows:
\begin{itemize}
\item
vertices $v_x$ correspond to partial solutions $x\in\D$;
\item
the root $r$ corresponds to the empty solution where no variables have been set;
\item
children of a vertex $v_x$ are the vertices $v_y$ corresponding to possible extensions $y$ of the partial solution $x$ that
$\A$ might try (for example, a backtracking algorithm for SAT   might choose one variable and try all possible values for
this variable), in the order in which $\A$ would try them.
\end{itemize}

{\bf Two-player games.} $\T$ may also be a position tree in a 2-player game, with $r$ corresponding to the current position.
Then, children of a node $v$ are the positions to which one could go by making a move at the position $v$.
A vertex $v$ is a leaf if we stop the evaluation at $v$ and do not evaluate children of $v$.

{\bf DAGs of unknown structure.} We also consider a generalization of this scenario to directed acyclic graphs (DAGs).
Let $\G$ be a directed acyclic graph. We assume that
\begin{itemize}
\item
Every vertex $v$ is reachable from the root $r$ via a directed path.
\item
The vertices of $\G$ can be divided into layers so that all edges from layer $i$ go to layer $i+1$.
	\item
	Given a vertex $v$, we can obtain the number $d(v)$ of edges $(u,v)$ and the number $d'(v)$ of edges $(v,u)$.
	\item
	Given a vertex $v$ and a number $i\in [d(v)]$, we can obtain the $i^{\rm th}$ vertex $u$ with an edge $(u,v)$.
	\item
	Given a vertex $v$ and a number $i\in [d'(v)]$, we can obtain the $i^{\rm th}$ vertex $u$ with an edge $(v,u)$.
\end{itemize}

\subsection{Notation}
By $  [a\tdots b] $, with $ a,b $ being integers, $ a\leq b  $, we denote the set $ \lrb{a,a+1,a+2, \ldots, b} $. When $ a=1 $, notation $ [a \tdots b] $ is simplified to $ [b] $.

We shall use the following notation for particular matrices:
\begin{itemize}
	\item $ \mb I_k $: the $ k \times k $ identity matrix;
	\item  $ \mb 0_{k_1,k_2} $: the $ k_1 \times k_2 $ all-zeros matrix.
\end{itemize}
We use the following notation for parameters describing a tree $\T$ or a DAG $\G$:
\begin{itemize}
\item
$T$ denotes the number of edges in $\T$ (or $\G$) or an upper bound on the number of edges which is given to an algorithm.
\item
$n$ denotes the depth of $\T$ (or $\G$) or an upper bound on the depth which is given to an algorithm.
\item
$d$ denotes the maximum possible total degree of a vertex $v\in\T$   ($\G$).
\item 
For any  vertex $ x \in \T $  (where $ \T $ is a tree),  the subtree rooted at  $x$ will be denoted by $ \T(x) $. 
\end{itemize}

\subsection{Eigenvalue estimation}

Quantum eigenvalue estimation is an algorithm which, given a quantum circuit implementing a
unitary $U$ and an eigenstate $ \ket {\psi} $ s.t.  $U\ket{\psi}=e^{i\theta}\ket{\psi}$, 
produces an estimate for $\theta$. It is known that one can produce 
an estimate $\hat{\theta}$ such that $Pr[|\theta-\hat{\theta}|\leq \delta_{est}]\geq 1-\epsilon_{est}$
with $O(\frac{1}{\delta_{est}} \log \frac{1}{\epsilon_{est}})$ repetitions of a circuit for controlled-$U$.

If eigenvalue estimation is applied to a quantum state $\ket{\psi}$ that is a superposition of several eigenstates 
\[ \ket{\psi} = \sum_{j} \alpha_j \ket{\psi_j}, 
\quad
U\ket{\psi_j} = e^{i\theta_j} \ket{\psi_j}, \]
the result is as if we are randomly choosing $j$ with probability $|\alpha_j|^2$ and estimating $\theta_j$.

In this paper, we use eigenvalue estimation to estimate the eigenvalue $e^{i\theta_{min}}$ that is closest to 1
(by that, here and later, we mean the eigenvalue which is closest to 1 among all  eigenvalues that are distinct from 1, i.e., the eigenvalue  $e^{i\theta_{min}}$ with the smallest nonzero absolute value $|\theta_{min}|$). We assume that we can produce a state $\ket{\psi_{start}}$ such that $\ket{\psi_{start}}$ is orthogonal to all 1-eigenvectors of $ U $ and
\[ |\lbra \Psi_+ |\psi_{start} \rket|^2 + |\lbra \Psi_- |\psi_{start} \rket|^2 \geq C \]
where $\ket{\Psi_+}$ and $\ket{\Psi_-}$ are eigenstates
with eigenvalues $e^{i \theta_{min}}$ and $e^{-i \theta_{min}}$ and $C$ is a known constant.
(If $U$ does not have an eigenvector with an eigenvalue $e^{-i\theta_{min}}$, the condition should be replaced by 
$|\lbra \Psi_+ |\psi_{start} \rket|^2\geq C$.) 
We claim
\begin{lemma}
\label{lem:unique}
Under the conditions above, there is an algorithm  which produces an estimate $\hat{\theta}$ such that 
$Pr[|\theta_{min}-\hat{\theta}|\leq \delta_{min}]\geq 1-\epsilon_{min}$
with 
\[ O\left(\frac{1}{C} \frac{1}{\delta_{min}} \log \frac{1}{C}  \log^2 \frac{1}{\epsilon_{min}} \right)\]
 repetitions of a circuit for controlled-$U$.
\end{lemma}

\begin{proof}
In Section \ref{app:ee}.
\end{proof}

\section{Results and algorithms}
\label{sec:results}

\subsection{Results on estimating sizes of trees and DAGs}

In this subsection, we consider the following task:

{\bf Tree size estimation.} The input data consist of a tree $\T$ and a value $T_0$ which is supposed to be an upper bound 
on the number of vertices in the tree. The algorithm must output an estimate for the size of the tree. The estimate
can be either a number $\hat{T}\in[T_0]$ or a claim ``$\T$ contains more than $T_0$ vertices". We say that
the estimate is $\delta$-correct if:
\begin{enumerate}
\item
the estimate is $\hat{T}\in[T_0]$ and it satisfies $|T-\hat{T}|\leq \delta T$ where $T$ is the actual number of vertices;
\item
the estimate is ``$\T$ contains more than $T_0$ vertices" and the actual number of vertices $T$ satisfies $(1+\delta) T> T_0$.
\end{enumerate}
We say that an algorithm solves the tree size estimation problem up to precision $ 1 \pm \delta $  with correctness probability at least $1-\epsilon$ if, for any $\T$ and any $T_0$, the probability that it outputs a $\delta$-correct estimate is at least $1-\epsilon$.

More generally, we can consider a similar task for DAGs.

{\bf DAG size estimation.}  
The input data consist of a directed acyclic graph $\G$ and a value $T_0$ which is supposed to be an upper bound on the 
number of edges in $\G$. The algorithm must output an estimate for the number of edges. The estimate
can be either a number $\hat{T}\in[T_0]$ or a claim ``$\G$ contains more than $T_0$ edges". We say that
the estimate is $\delta$-correct if:
\begin{enumerate}
	\item
	the estimate is $\hat{T}\in[T_0]$ and it satisfies $|T-\hat{T}|\leq \delta T$ where $T$ is the actual number of edges;
	\item
	the estimate is ``$\G$ contains more than $T_0$ edges" and the actual number of edges $T$ satisfies $(1+\delta) T> T_0$.
\end{enumerate}
We say that an algorithm solves the DAG size estimation problem up to precision $ 1 \pm \delta $  with correctness probability at least $1-\epsilon$ if, for any $\G$ and any $T_0$, 
the probability that it outputs a $\delta$-correct estimate is at least $1-\epsilon$.

Tree size estimation is a particular case of this problem: since a tree with  $T$ edges has $T+1$ vertices, estimating the number of vertices and the number of edges
are essentially equivalent for trees. We show

\begin{theorem}
	\label{thm:main-DAG}
	DAG size estimation up to precision $1\pm \delta$ can be solved with the correctness probability at least $1-\epsilon$
	by a quantum algorithm which makes 
	\[ 
	O\left(
	\frac{\sqrt{nT_0}}{\delta^{1.5}}  d \log^2 \frac{1}{\epsilon} 
	\right) 
	\]
queries to black boxes specifying $\G$ and $O(\log T_0)$ non-query transformations per query.
\end{theorem}

{\bf Note.} If we use $\tilde{O}$-notation, the $O(\log T_0)$ factor can be subsumed into the $\tilde{O}$ 
and the time complexity is similar to query complexity.

\subsection{Algorithm for DAG size estimation}\label{sec:GSE}

In this subsection, we describe the algorithm of Theorem \ref{thm:main-DAG}.
The basic framework of the algorithm (the state space and the transformations that we use) is adapted from Montanaro \cite{Montanaro2015}.

Let $ \mc G  = (\vertSet, \edgeSet) $ be a directed acyclic graph, with $ \lrv {\vertSet} =\noTot $ vertices and $ \lrv{\edgeSet} = T $ edges.    We assume that   the root is labeled as $ v_1   $.

For each vertex $ v_i \in \vertSet $
\begin{itemize}
	\item  $ \ell(i) \leq n$ stands for the distance from $ v_i $ to the root,   
	\item $ d_i \leq d$ stands for the total degree of the vertex $ v_i $. 
In notation in Section  \ref{sec:setting}, we have $ d_i =   d(v_i) + d'(v_i) $.
\end{itemize}
For technical purposes we also introduce an additional vertex $ \addV $ and an additional edge $\addE =  (\addV , v_1) $ which connects $ \addV $ to the root. Let $\complV = \vertSet \cup \lrb{\addV}  $, $ \complE=\edgeSet \cup \lrb {\addE} $ and $ \complG = (\complV,\complE) $.

For each vertex $ v $ by $ \nhood(v) $ we denote    the set of all edges in $ \edgeSet $ incident to $ v $ (in particular, when $ v =v_1$ is the root, the additional edge $ \addE \notin \edgeSet$ is not included in $ \nhood(v_1) $).

Let $ \vertSetA$ be the set of vertices at an even distance from the root (including the root itself) and $ \vertSetB $ be the set of vertices at an odd distance from the root. Let $ \noA = \lrv {\vertSetA} $ and $ \noB = \lrv {\vertSetB}$, then  $\noTot = \noA + \noB $.  Label the vertices in $ \vertSet $ so that  $\vertSetA= \lrb{v_1, v_2,\ldots,v_{\noA}}$ and $ \vertSetB=\lrb{v_{\noA+1}, v_{\noA+2}, \ldots, v_{\noA+\noB}}$.

Let  $ \alpha > 0 $ be fixed.  
Define a Hilbert space $ \mc H $ spanned by $ \lrb{\ket e\ \vline\  e\in  \complE} $ (one basis state per edge, including the additional edge). 
For each vertex $ v \in  \vertSet $  define a vector $ \ket {s_v} \in \mc H$ as
\[ 
\ket {s_v} = 
\begin{cases}
\ket{\addE} + \alpha \sum_{e \in \nhood(v)}  \ket e, & v= v_1\\
\sum_{e \in \nhood(v)}  \ket e, & v \neq  v_1.
\end{cases}
\]
Define a subspace $ \spaceA \subset \mc H $ spanned by $  \lrb{ \ket {s_v} \ \vline\ v \in \vertSetA } $ and a subspace $ \spaceB \subset \mc H $ spanned by $  \lrb{  \ket {s_v} \ \vline\  v \in \vertSetB } $. 
Define a unitary operator  $ \RA $  which  negates all vectors in $ \spaceA $ (i.e., maps $ \ket \psi $ to $ -\ket\psi $ for all $ \ket\psi \in \spaceA $)  and leaves $ \spaceA^\perp $   invariant. Analogously,  a unitary operator  $ \RB $     negates all vectors in $ \spaceB $    and leaves $ \spaceB^\perp $   invariant. 

Similarly as in \cite{Montanaro2015}, both $ \RA $ and $ \RB $ are implemented as the direct sum of  diffusion operators $ D_v $.
Let a subspace $ \mc H_v $, $ v \in \vertSet $, be spanned  by $ \lrb{  \ket e \ \vline\   e \in \nhood(v)} $ (or  $ \lrb{  \ket e \ \vline\   e \in \nhood(v)} \cup \lrb{\ket \addE} $ when $ v=v_1 $ is the root). Define the diffusion operator  $ D_v $, which  acts on the subspace $ \mc H_v $, as $ I - \frac{2}{\nrm{s_v}^2} \ket {s_v}\bra  {s_v} $. This way, each $ D_v $ can be implemented with only knowledge of $ v $ and its neighborhood. (A minor difference from \cite{Montanaro2015}: since we are concerned with tree size estimation problem now, we make no assumptions about any vertices being marked at this point and therefore $ D_v $ is not the identity for any $ v \in \vertSet $.)

Then, 
\[ 
\RA = \bigoplus_{v \in \vertSetA} D_v 
\quad\text{and}\quad
\RB = \ket {\addE}\bra  {\addE}  + \bigoplus_{v \in \vertSetB} D_v.
 \]
In Section \ref{sec:efficient}, we show

\begin{lemma}
\label{lem:ra}
Transformations $\RA$ and $\RB$ can be implemented using $O(d)$ queries and $O(d \log V)$ non-query gates. 
\end{lemma}

We note that $ \RA $ and $ \RB $ are defined with respect to a parameter $ \alpha $, to be specified in the algorithm that uses the transformations $\RA$ and $\RB$.

The algorithm of Theorem \ref{thm:main-DAG}  for estimating size of DAGs is as follows:

\begin{algorithm}[tbhp]
\caption{Algorithm for DAG size estimation}
\label{alg:dag}
\begin{enumerate}
\item
Apply the algorithm of Lemma \ref{lem:unique} for the transformation $\RB \RA$ (with  $ \alpha = \sqrt{2n \delta^{-1}}$) with the state $\ket{\psi_{start}}=\ket{\addE}$ and  parameters $C=\frac{4}{9}$, $\epsilon_{min}=\epsilon$, 
$\delta_{min}=\frac{\delta^{1.5}}{ 24\sqrt { 3nT_0 }}$
\item
Output $\hat{T} = \frac{1}{\alpha^2 \sin^2 \frac{\hat \theta}{2}}$ as the estimate for the number of edges.
\end{enumerate}
\end{algorithm}

\subsection{Analysis of Algorithm \ref{alg:dag}}

We now sketch the main ideas of the analysis of Algorithm \ref{alg:dag}.
From Lemma \ref{th:1eigenvectors} in Section \ref{sec:1eigen} 
it follows that $\RB \RA$  has no 1-eigenvector $\ket{\psi}$ with $\lbra \psi | \addE \rket\neq 0$.
	Let $\ket{\Psi_+}$ and $\ket{\Psi_-}$ be the two eigenvectors of $\RB \RA$ with eigenvalues $e^{\pm i \theta}$ closest to 1.
	Lemma \ref{lem:prob-success} shows that the starting state $\ket{\addE}$ has sufficient overlap with the subspace spanned by these two vectors for applying the algorithm of Lemma \ref{lem:unique}.
	\begin{lemma}
		\label{lem:prob-success}
		If $ \alpha\geq \sqrt {2n} $, we have
		\[
		 \braket{\addE}{q_2}  \geq    \frac{2}{3}
		  \]
		for a state $\ket{q_2}\in  \Span \lrb{\ket {\Psi_+}, \ket {\Psi_-}} $.
	\end{lemma}
	
	\begin{proof}
		In Section \ref{sec:phase}.
	\end{proof}
	
	Lemma \ref{lem:theta} shows that the estimate $\hat \theta$ provides a good estimate $\hat T$ for the size of the DAG.
	\begin{lemma}
		\label{lem:theta}
		Suppose that  $ \delta \in (0,1) $. Let    $ \alpha = \sqrt{2 n  \delta^{-1}} $   and $ \hat \theta \in  (0;\pi/2)$ satisfy
		\[ 
		\lrv{\hat \theta - \theta } \leq \frac{\delta^{1.5}}{ {24} \sqrt { 3nT }}.
		\]
		Then
		\[ 
		(1 - \delta) T \leq   \frac{1}{\alpha^2 \sin^2 \frac{\hat \theta }{2}}  \leq  (1+ \delta)T.
		\]
	\end{lemma}
	
	\begin{proof}
		In Section \ref{sec:phase}.
	\end{proof}

Lemmas \ref{lem:unique}, \ref{lem:prob-success} and \ref{lem:theta} together imply that Algorithm \ref{alg:dag}
outputs a good estimate with probability at least $1-\epsilon$. According to Lemma \ref{lem:unique}, we need to invoke
controlled versions of $\RA$ and $\RB$
\[ O\left(\frac{1}{C} \frac{1}{\delta_{min}} \log \frac{1}{C}  \log^2 \frac{1}{\epsilon_{min}} \right) =
O\left( \frac{\sqrt{n T_0}}{\delta^{1.5}} \log^2 \frac{1}{\epsilon} \right) \]
times and because of Lemma \ref{lem:ra}, each of these transformations can be 
performed with $O(d)$ queries and $O(d \log V) = O(d \log T_0)$ non-query transformations.
	
Proof of Lemmas \ref{lem:prob-success} and \ref{lem:theta} consists of a number of steps:
\begin{enumerate}
\item
We first relate the eigenvalues of $\RA \RB$ to singular values of an $\noA\times \noB$ matrix $L$.
The matrix $L$ is defined by $L[v, w] = \frac{\lbra s_v \ket{s_w}}{\|s_v\| \cdot \|s_w\|}$.
Because of  a correspondence by Szegedy \cite{Szegedy2004b}, a pair of eigenvalues $e^{\pm i\theta}$ of $\RB \RA$ corresponds 
to a singular value $\lambda=\cos \frac{\theta}{2}$ of $L$.
\item
Instead of $L$, we consider $K=(I-L L^*)^{-1}$ (with both rows and columns indexed by elements of $\vertSetA$). 
A singular value $\lambda$ of $L$ corresponds to an eigenvalue $  (1-\lambda^2)^{-1}$ of $K$.
\item
We relate $K$ to the fundamental matrix $N$ of a certain classical random walk on the graph $\mc G$.
(The entries of the fundamental matrix $N[i, j]$ are the expected number of visits to $j$ that the random walk makes 
if it is started in the vertex $i$.)
\item
We relate $N$ to the resistance between $i$ and $j$ if the graph $\mc G$ is viewed as an electric network.
\item
We bound the electric resistance, using the fact that the resistance only increases if an edge is removed form $\mc G$. Thus, 
the maximum resistance is achieved if $\mc G$ is a tree.
\end{enumerate}

This analysis yields that the entries of $K$ can be characterized by the inequalities
\[ 
\alpha^2     a[i] a[j]
\leq K[i,j] \leq  
\lr{\alpha^2    + n} a[i] a[j]
\]
where 
$ a[1] = \sqrt{d_1 +  \alpha^{-2} }$ and $ a[j] = \sqrt{d_j}$ for $  j \in [2 \tdots A] $
(Lemma \ref{lem:corr9} in  Section \ref{sec:lem13}). 
From this we derive bounds on the largest eigenvalue of $K$  which imply bounds on the eigenvalue of $\RB \RA$ that is closest to 1.

We describe the analysis in more detail in Section \ref{sec:dag-analysis}.

\subsection{Better backtracking algorithm}

{\bf Backtracking task.}
We are given a tree $\T$ and a black-box function 
\[ P:V(\T)\rightarrow \lrb{true, false,indeterminate} \]
(with $P(x)$ telling us whether $x$ is a solution to the computational problem we are trying to solve), where $ V(\T) $ stands for the set of vertices of $ \T $ and $ P(v) \in  \lrb{true, false}$ iff $ v $ is a leaf. A vertex $ v\in V(\T) $ is called marked if $ P(v)=true $.
We have to determine whether $\T$ contains a marked vertex.

For this section, we assume that the tree is binary. 
(A vertex with $ d $ children can be replaced by a binary tree of depth $ \lceil \log d \rceil $. This increases the size of the tree by a constant factor, the depth by a factor of at most $ \lceil \log d \rceil $ and the complexity bounds by a polylogarithmic factor of $ d $.)

\begin{theorem}
\cite{Montanaro2015}
\label{thm:montanaro}
There is a quantum algorithm which,
given
\begin{itemize}
\item
a tree $\T$ (accessible through black boxes, as described in Section \ref{sec:setting}), 
\item
an access to the  black-box function $P$, and 
\item
numbers $T_1$ and $n$ which are upper bounds on the size and the depth of $\T$,
\end{itemize}
determines if the tree contains a marked vertex with query and time complexity 
$O(\sqrt{T_1 n} \log \frac{1}{\epsilon})$, 
with the probability of a correct answer at least $1-\epsilon$.
\end{theorem}

The weakness of this theorem is that the complexity of the algorithm depends on $T_1$. On the other hand, a classical
backtracking algorithm $\A$ might find a solution in substantially less than $T_1$ steps (either because 
the tree $\T$ contains multiple vertices $x:P(x)=true$ or because the heuristics that $\A$ uses to decide which
branches to search first are likely to lead to $x:P(x)=true$). 

We improve on Montanaro's algorithm by showing

\begin{theorem}
\label{thm:search}
Let $\A$ be a classical backtracking algorithm. 
There is a quantum algorithm that, with probability at least $1-\epsilon$, outputs 1 if 
$\T$ contains a marked vertex and 0 if $\T$ does not contain a marked vertex and uses
\[ 
O\left(
n\sqrt{nT}   \log^2 \frac{n \log T_1}{\epsilon} 
\right) 
\]
queries and $O(\log T_1)$ non-query transformations per query where 
\begin{itemize}
\item
$T_1$ is an upper bound on the size of $\T$ (which is given to the quantum algorithm),
\item
$n$ is an upper bound on the depth of the $\T$ (also given to the quantum algorithm), 
\item
$T$ is the number of vertices of $\T$ actually explored by $\A$.
\end{itemize}
\end{theorem}

\begin{proof}
The main idea of our search algorithm is to generate subtrees of $\T$ that consist of first approximately $2^i$ vertices visited by
the classical backtracking strategy $\A$, increasing $i$ until a marked vertex is found or until we have searched the whole tree $\T$.

Let $\T_{m}$ be the subtree of $\T$ consisting of the first $m$ vertices visited by the classical backtracking algorithm $\A$.
Then, we can describe $T_m$ by giving a path 
\begin{equation}
\label{eq:path}
 r = u_0\rightarrow u_1 \rightarrow u_2 \rightarrow \ldots \rightarrow u_l = u 
\end{equation}
where $u$ is the $m^{\rm th}$ vertex visited by $\A$. Then, $\T_m$ consists of all the subtrees $\T(u)$ rooted at $u$ 
such that $u$ is a child of $u_i$ (for some $i \in [0 \tdots l-1]$) that is visited before $u_{i+1}$ and the vertices
$u_0, \ldots, u_l$ on the path. 

Given access to $\T$ and the path (\ref{eq:path}), one can simulate Montanaro's algorithm on $\T_m$. 
Montanaro's algorithm consists of performing the transformations similar to $\RA$ and $\RB$ described in Section 
\ref{sec:GSE}, except that $ D_v $ is identity  if $ v $ is marked. To run Montanaro's algorithm on $\T_m$, 
we use access to $\T$ but modify the transformations of the algorithm as follows:
\begin{itemize}
\item
when performing $D_v$ for some $v$, we check if $v$ is one of vertices $u_i$ on the path;
\item
if $v=u_i$ for some $i\in[0 \tdots l-1]$ and $u_{i+1}$ is the first child of $u_i$, we change $D_v$ as if $u_{i+1}$ is the only child
of $u_i$;
\item
otherwise, we perform $D_v$ as usually.
\end{itemize} 

\begin{lemma}
\label{thm:gen-path}
There is a quantum algorithm that generates the path
\[ r = u_0\rightarrow u_1 \rightarrow u_2 \rightarrow \ldots \rightarrow u_l = u \]
corresponding to a subtree $\T_{\hat{m}}$ for $\hat m:|m-\hat{m}|\leq \delta m$ with
a probability at least $1-\epsilon$ and uses
\[ 
O\left(
\frac{n^{1.5}\sqrt{m}}{\delta^{1.5}}  \log^2 \frac{n}{\epsilon} 
\right) 
\]
queries and $O(\log T_0)$ non-query transformations per query.
\end{lemma}

\begin{proof}
The algorithm {\bf Generate-path} is described as Algorithm \ref{alg:gen}.

\begin{algorithm}[tbhp]
\caption{Procedure {\bf Generate-path($v, m$)}}
\label{alg:gen}

{\bf Generate-path($v, m$)} returns a path (\ref{eq:path}) defining a subtree $\T_{\hat m}$, 
with $\hat m$ satisfying $|m-\hat{m}|\leq \delta m$ with a probability at least $1-\epsilon$.

\begin{enumerate}
\item
If $v$ is a leaf, return the empty path.
\item
Otherwise, let $v_1, v_2$  be the children of $v$, in the order in which $\A$ visits them.
\item
Let $m_1$ be an estimate for the size of $\T(v_1)$, using the algorithm for  the tree size estimation with
the precision $1\pm \delta$, the probability of a correct answer at least $1-\frac{\epsilon}{n}$ and 
$ \frac{m-1}{1-\delta} $ as the upper bound on the tree size.
\item
If $m_1> m-1$, return the path obtained by concatenating the edge $v\rightarrow v_1$ with the path returned by {\bf Generate-path($v_1, m-1$)}.
\item
If $m_1= m-1$, return the path obtained by concatenating the edge $v\rightarrow v_1$ with the path from $v_i$ 
to the last vertex of $\T(v_i)$ (that is, the path in which we start at $v_i$ and, at each vertex, choose the child that is 
the last in the order in which $\A$ visits the vertices).
\item
If $m_1< m-1$, return the path obtained by concatenating the edge $v\rightarrow v_2$ with the path returned by {\bf Generate-path($v_2, m-1-m_1$)}.
\end{enumerate} 
\end{algorithm}

{\bf Correctness.}
{\bf Generate-path($v, m$)} invokes the tree size estimation once 
and may call itself recursively once, with $v_1$ or $v_2$ instead of $v$.
Since the depth of the tree is at most $n$, the depth of the recursion is also 
at most $n$. On all levels of recursion together, there are at 
most $n$ calls to tree size estimation. If we make the probability of error for tree size estimation at most 
$\frac{\epsilon}{n}$, the probability that all tree size estimations return sufficiently precise estimates is at least $1-\epsilon$.
Under this assumption, the number of vertices in each subtree added to $\T'$ is within a factor of $1\pm \delta$ of the estimate.
This means that the total number of vertices in $\T'$ is within $1\pm\delta$ of $m$ (which is equal to the sum of estimates).

{\bf Query complexity.}
Tree size estimation is called at most $n$ times, with the complexity of 
\[ 
O\left(\frac{\sqrt{nm}}{\delta^{1.5}}   \log^2 \frac{n}{\epsilon} \right)
 \]
each time, according to Theorem \ref{thm:main-DAG}. Multiplying this complexity by $n$ gives 
Lemma \ref{thm:gen-path}.
\end{proof}

We now continue with the main algorithm for Theorem \ref{thm:search} (Algorithm \ref{alg:2}). 

\begin{algorithm}[tbhp]
\caption{Main part of the quantum algorithm for speeding up backtracking}
\label{alg:2}

\begin{enumerate}
\item Let $i=1$.
\item Repeat:
\begin{enumerate}
\item Run {\bf Generate-path($r, 2^i$)} with $\delta=\frac{1}{2}$ and error probability  at most $\frac{\epsilon}{2 \lceil \log T_1 \rceil}$, obtain a path 
defining a tree $\T'=\T_{\hat m}$.
\item Run Montanaro's algorithm on $\T'$, with the upper bound on the number of vertices
$\frac{3}{2} 2^i$ and the error probability at most $\frac{\epsilon}{2 \lceil \log T_1 \rceil}$, 
stop if a marked vertex is found.
\item Let $i=i+1$.
\end{enumerate}
until a marked vertex is found or $\T'$ contains the whole tree.
\item
If $\T'$ contains the whole tree and no marked vertex was found in the last run, stop.
\end{enumerate}
\end{algorithm}

{\bf Correctness.}
Each of the two subroutines ({\bf Generate-path} and Montanaro's algorithm) is invoked at most $\lceil \log T_1 \rceil$ times.
Hence, the probability that all invocations are correct is at least $1-\epsilon$.

{\bf Query complexity.}
By Lemma \ref{thm:gen-path}, the number of queries performed by 
{\bf Generate-path} in the $i^{\rm th}$ stage of the algorithm is 
\[ 
O\left(
n^{1.5} \sqrt{2^i}    \log^2 \frac{n \log T_1}{\epsilon} 
\right) 
\]
and the complexity of Montanaro's algorithm in the same stage of the algorithm is of a smaller order.
Summing over $i$ from 1 to $\lceil\log \frac{T}{1-\delta} \rceil$ (which is the maximum possible value of $i$ if 
all the subroutines are correct) gives the query complexity
\[ 
O\left(
n^{1.5} \sqrt{T}   \log^2 \frac{n \log T_1}{\epsilon} 
\right) .
\]
\end{proof}

{\bf Note.}  If $T$ is close to the size of the entire tree, the complexity of Algorithm \ref{alg:2} (given by Theorem \ref{thm:search}) 
may be   larger than the complexity of Montanaro's algorithm (given by Theorem \ref{thm:montanaro}).
To deal with this case, one can stop Algorithm \ref{alg:2} when the number of queries exceeds the expression in Theorem \ref{thm:montanaro} and then run Montanaro's algorithm on the whole tree.
Then, the complexity of the resulting algorithm is the minimum of complexities in Theorems \ref{thm:montanaro} and
\ref{thm:search}.

\subsection{Evaluating AND-OR formulas of unknown structure}
\label{sec:andor}

We now consider evaluating AND-OR formulas in a similar model where we are given the root of the formula
and can discover the formula by exploring it locally. This corresponds to position trees in 2-player games where we know
the starting position (the root of the tree) and, given a position, we can generate  all possible positions after one move.
More precisely, we assume access to
\begin{itemize}
\item
a formula tree $\T$ (in the form described in Section \ref{sec:setting});
\item 
a black box which, given an internal node, answers whether AND or OR should be evaluated at this node;
\item
a black box which, given a leaf, answers whether the variable at this leaf is 0 or 1.
\end{itemize} 

\begin{theorem}
\label{th:andor}
There is a quantum algorithm which evaluates an AND-OR tree 
of size at most $T$ and depth $n=T^{o(1)}$ in this model running in time $O(T^{1/2+\delta})$,
for an arbitrary $\delta>0$.
\end{theorem}

\begin{proof}
We assume that the tree is binary. (An AND/OR node with $k$ inputs can be replaced by a binary tree of depth $\lceil \log k\rceil$
consisting of gates of the same type. This increases the size of the tree by a constant factor and the depth by a factor of at most 
$\lceil \log k \rceil$.)

We say that $\T'$ is an $m$-heavy element subtree of $\T$ if it satisfies the following properties:
\begin{enumerate}
\item
$\T'$ contains all $x$ with $|\T(x)|\geq m$ and all children of such $x$;
\item
all vertices in $\T'$ are either $x$ with  $|\T(x)|\geq \frac{m}{2}$ or children of such $x$.
\end{enumerate}

\begin{lemma}
\label{lem:heavy}
Let $\T$ be a tree and let $T$ be an upper bound on the size of $\T$.
There is a 
\[ O\left( \frac{n^{1.5} T}{\sqrt{m}} \log^2 \frac{T}{\epsilon} \right) \]
time quantum algorithm that generates $\T'$ such that, with probability at least $1-\epsilon$, $\T'$ is an $m$-heavy element subtree of
$\T$.
\end{lemma}

\begin{proof}
The algorithm is as follows.

\begin{algorithm}[tbhp]
\caption{Algorithm {\bf Heavy-subtree}$(r,m,\epsilon)$}
\label{alg:3}

\begin{enumerate}
\item 
 Run the tree size estimation for $\T(r)$ with $m$ as the upper bound on the number of vertices and parameters
$\delta=\frac{1}{4}$ and $\epsilon' = \frac{m}{6nT} \epsilon$.
\item
If the estimate is smaller than $\frac{2m}{3}$, return $\T'$ consisting of the root $r$ only.
\item 
Otherwise, let $\T'=\{r\}$. Let $v_1$ and $v_2$ be the children of $r$. For each $i$,  invoke {\bf Heavy-subtree}$(v_i, m, \epsilon)$ recursively and add all vertices from the subtree returned
by {\bf Heavy-subtree}$(v_i, m, \epsilon)$ to $\T'$. Return $\T'$ as the result.
\item
If, at some point, the number of vertices added to $\T'$ reaches $\frac{6T}{m}n$, stop and return the current $\T'$.
\end{enumerate}
\end{algorithm}

The proof of correctness and the complexity bounds are given in Section \ref{app:andor}.
\end{proof}

To evaluate an AND-OR formula with an unknown structure, let $T$ be an upper bound on the size of the formula $F$. Let $c$ be an integer. For $i=1, \ldots, c$, we define $T_i = T^{i/c}$.  To evaluate $F$, we identify an $T_{c-1}$-heavy element subtree
$F'$ and then run the algorithm of \cite{A+} or \cite{R} for evaluating formulas with a known structure on it.
The leaves of this $F'$ are roots of subtrees of size at most $T_{c-1}$. To perform queries on them, 
we call the same algorithm recursively. That is, given a leaf $v$, we identify a $T_{c-2}$-heavy element subtree
$F'(v)$ and then run the algorithm of \cite{A+} or \cite{R} 
for evaluating formulas with a known structure on it, with queries at the leaves
replaced by another recursive call of the same algorithm, now on a subtree of size at most $T_{c-2}$.

The algorithm that is being called recursively is described as Algorithm \ref{alg:4}. To evaluate the original formula $F$,
we call {\bf Unknown-evaluate}($r, c, \epsilon$).

\begin{algorithm}[tbhp]
\caption{Algorithm {\bf Unknown-evaluate}($r, i, \epsilon$)}
\label{alg:4}

\begin{enumerate}
\item
If $i=1$, determine the structure of the tree by exploring it recursively, let $\T'$ be the resulting tree.
\item
If $i>1$, use {\bf Heavy-subtree}$(r, T_{i-1}, \epsilon/5)$ to obtain $\T'$. Let $s$ be the size of $\T'$.
\item
Run the AND-OR formula evaluation algorithm for known formulas of \cite{A+} to evaluate the formula 
corresponding to $\T'$, with a probability of a correct answer at least $1-\frac{\epsilon}{5}$. If $i>1$,
use calls to {\bf Unknown-evaluate}($v, {i-1}, \epsilon/s^3$) instead of queries at leaves $v$.
\item
If {\bf Unknown-evaluate} is used as a query at a higher level (that is, if $i<c$), 
perform a phase flip to simulate the query and run the first three steps in reverse, 
erasing all the information that was obtained during this execution of {\bf Unknown-evaluate}.
\end{enumerate}
\end{algorithm}

The proof of correctness and the complexity bounds are given in Section \ref{app:andor}.
\end{proof}

\section{Proof of Lemma \ref{lem:unique}}\label{app:ee}

We perform ordinary eigenvalue estimation $t=\lceil \frac{1}{C} \ln \frac{2}{\epsilon_{min}} \rceil$ times, with
parameters $\delta_{est}=\delta_{min}$ and $\epsilon_{est} = \frac{\epsilon_{min}}{2t}$. We then take 
\[ \hat{\theta} = \min (|\hat{\theta_1}|, \ldots, |\hat{\theta_t}|) \]
where $\hat{\theta_1}, \ldots, \hat{\theta_t}$ are the estimates that have been obtained.  

To see that $Pr[|\theta_{min}-\hat{\theta}|\leq \delta_{min}]\geq 1-\epsilon_{min}$, we observe that:
\begin{enumerate}
\item
The probability that none of $\hat{\theta_j}$ is an estimate for $\pm\theta_{min}$ is at most
\[ (1-C)^t \leq e^{-\ln\frac{2}{\epsilon_{min}}} = \frac{\epsilon_{min}}{2} .\]
\item
The probability that one or more of $\hat{\theta_j}$ differs from the corresponding $\theta_j$ by more than
$\delta_{est}$ is at most $t \epsilon_{est} = \frac{\epsilon_{min}}{2}$.
\end{enumerate}
If none of these two ``bad events" happens, we know that, among $\hat{\theta_j}$, 
there is an estimate for $\pm\theta_{min}$ that differs from $\theta_{min}$ or $-\theta_{min}$ by at most $\delta_{est}$. 
Moreover, any estimate $\hat{\theta_j}$ 
for $\theta_j\neq \pm\theta_{min}$ must be at least $|\theta_j|-\delta_{est} \geq \theta_{min}-\delta_{est}$ in absolute value.
Therefore, even if $\hat{\theta_j}$ with the smallest $|\hat{\theta_j}|$ is not an estimate for $\pm\theta_{min}$,
it must still be in the interval $[\theta_{min}-\delta_{est}, \theta_{min}+\delta_{est}]$.

The number of repetitions of controlled-$U$ is
\[ O\left( t\frac{1}{\delta_{est}} \log \frac{1}{\epsilon_{est}} \right) = 
O\left( \frac{1}{C} \left( \log \frac{1}{\epsilon_{min}} \right)
\frac{1}{\delta_{min}}\log \frac{1}{\epsilon_{est}} \right) \]
and we have $\log \frac{1}{\epsilon_{est}} = \log t + \log \frac{1}{\epsilon_{min}}+ O(1)$. Also,
$\log t \leq \log \frac{1}{C} + \log \log \frac{1}{\epsilon_{min}} + O(1)$. Therefore,
\[ \log \frac{1}{\epsilon_{est}} \leq \log \frac{1}{C} + O\left(\log \frac{1}{\epsilon_{min}}\right) 
= O\left(\log \frac{1}{C} \log \frac{1}{\epsilon_{min}} \right) ,\]
implying the lemma.

\section{Implementing transformations \texorpdfstring{$\RA$}{RA} and \texorpdfstring{$\RB$}{RB}}
\label{sec:efficient}

\newcommand{\HU}{{\mc H}_{\mathbf 1}} 
\newcommand{\HV}{{\mc H}_{\mathbf 2}} 
\newcommand{\HW}{{\mc H}_{\mathbf 3}} 

We represent the basis states $\ket{e}$ as $\ket{u, w}$ where $u\in \vertSetA$ and $w\in \vertSetB \cup \{v_{V+1}\}$.
Let $\HU$ and $\HV$ be the registers holding $u$ and $w$, respectively.

Each $\ket{s_u}$, $u\in \vertSetA$, can be expressed as $\ket{s_u} = \ket{u} \otimes \ket{s'_u}$.
We can view $\RA$ as a transformation in which, for each $u$ in $\HU$, we perform 
$D'_u= I - \frac{2}{\nrm{s'_u}^2} \ket{s'_u}\bra{s'_u}$ on the subspace $\ket{u}\otimes\HV$.

The transformation $D'_u$ can be performed as follows:
\begin{enumerate}
	\item
	Use queries to obtain the numbers of incoming and outgoing edges $d(u)$ and $d'(u)$
	(denoted by $d_{in}$ and $d_{out}$ from now on). Use queries to obtain vertices $w_1, \ldots, w_{d_{in}}$ with
	edges $(w_j, u)$ and vertices $w'_1, \ldots, w'_{d_{out}}$ with edges $(u, w_j)$.
	\item
	Let $\HW$ be a register with basis states $\ket{0}, \ldots, \ket{d+1}$.
	Use the information from the first step  to perform a map on $\HV\otimes \HW$ 
	that maps $\ket{w_1}\ket{0}$, $\ldots$, $\ket{w_{d_{in}}}\ket{0}$, $\ket{w'_1}\ket{0}$, \
	$\ldots$, $\ket{w'_{d_{out}}}\ket{0}$
	to $\ket{0}\ket{1}$, $\ldots$, $\ket{0}\ket{d_{in}+d_{out}}$ and, 
	if $u=v_1$, also maps $\ket{v_{V+1}}\ket{0}$ to $\ket{0}\ket{d_{in}+d_{out}+1}$.
	\item
	Perform the transformation $D=I-2\ket{\psi}\bra{\psi}$ on $\HW$ where 
	$\ket{\psi}$ is the state obtained by normalizing $\sum_{i=1}^{d_{in}+d_{out}} \ket{i}$ if $u\neq v_1$ and
	by normalizing $\sum_{i=1}^{d_{in}+d_{out}} \ket{i} + \alpha \ket{d_{in}+d_{out}+1}$ if $u=v_1$.
	\item
	Perform steps 2 and 1 in reverse.
\end{enumerate}
The first step consists of at most $d+2$ queries (two queries to obtain $d(u)$ and $d'(u)$ and at most $d$ queries to obtain
the vertices $w_i$ and $w'_i$) and some simple operations between queries, to keep count of vertices $w_i$ or $w'_i$
that are being queried.

For the second step, let ${\mc H}_j$ be the register holding the value of $w_j$ obtained in the first step.
For each $j\in\{1, \ldots, d_{in}\}$, we perform a unitary $U_j$ on ${\mc H}_j\otimes \HV \otimes \HW$ 
that maps $\ket{w_j}\ket{w_j}\ket{0}$ to $\ket{w_j}\ket{0}\ket{j}$. This can be done as follows:
\begin{enumerate}
	\item
	Perform $\ket{x}\ket{y} \rightarrow \ket{x}\ket{x\oplus y}$ on ${\mc H}_j\otimes \HV$ where $\oplus$ denotes bitwise XOR. 
	\item
	Conditional on the second register being $ \ket{0} $, add $j$ to the third register.
	\item
	Conditional on  the third register not being $ \ket{j} $, perform $\ket{x}\ket{y} \rightarrow \ket{x}\ket{x\oplus y}$ on ${\mc H}_j\otimes \HV$  to
	reverse the first step.
\end{enumerate}
We then perform similar unitaries $U_{d_{in}+j}$ that map $\ket{w'_j} \ket{w'_j} \ket{0}$ to $\ket{w'_j} \ket{0} \ket{d_{in}+j}$
and, if $u=v_1$, we also perform a unitary $U_{d_{in}+d_{out}+1}$ that maps $\ket{v_{V+1}}\ket{0}$ to $\ket{0}\ket{d_{in}+d_{out}+1}$.
Each of these unitaries requires $O(\log V)$ quantum gates and there are $O(d)$ of them.
Thus, the overall complexity is $O(d \log V)$.

The third step is a unitary $D$ on $\HW$ that depends on $d_{in}+d_{out}$ and on whether we have $u=v_1$. 
We can express $D$ as $D=U_{\psi} (I- 2\ket{0}\bra{0}) U^{-1}_{\psi}$ where $U_{\psi}$ is any transformation
with $U_{\psi}\ket{0} = \ket{\psi}$. Both $U_{\psi}$ and $I- 2\ket{0}\bra{0}$ are simple transformations
on a $\lceil \log(d+2) \rceil$ qubit register $\HW$ which can be performed with $O(\log d)$ gates. 

Since there are $d$ possible values for $d_{in}+d_{out}$ and 2 possibilities for whether $u=v_1$, 
we can try all possibilities one after another, checking the conditions and then performing 
the required unitary, if necessary, with $O(d \log d)$ gates.
Checking $u=v_1$ requires $O(\log V)$ gates.

The fourth step is the reverse of the first two steps and can be performed with the same complexity.

The overall complexity is $O(d)$ queries and $O(d \log V)$ quantum gates. The transformation $\RB$ can be performed
with a similar complexity in a similar way.

\section{Analysis of algorithm for DAG size estimation}
\label{sec:dag-analysis}

This section is devoted to the analysis of Algorithm \ref{alg:dag}.

\subsection{Spectral theorem}
\label{sec:Ra}

Our bounds on eigenvalues of $\RB \RA$ are based on  the spectral theorem from \cite{Szegedy2004b}.

Suppose that $ \mc X $ is a subspace of a Hilbert space $ \mc H $. Let $  \refX{\mc  X} $ denote a reflection which leaves $ \mc X $ invariant and negates all vectors in $ \mc X^\perp $.

Let $ Gram(w_1, \ldots, w_k) $ stand for the Gram matrix  of vectors $ w_1, \ldots, w_k  $, i.e., the matrix formed by the inner products $   \lra{w_s,w_t} $.
\begin{definition}[{\cite[Definition 5]{Szegedy2004b}}]
	Let $ \lr{ \lrb{w_1, \ldots, w_k}, \lrb{\tilde w_1, \ldots, \tilde w_l}  } $  be an ordered pair  of orthonormal systems.
	The discriminant matrix of this pair is 
	\[ 
	M = 
	Gram(w_1, \ldots, w_k, \tilde w_1, \ldots, \tilde w_l)   - I.
	\]
\end{definition}

Let   $ \mc A $ and $ \mc B $ be two subspaces of the same    
Hilbert space $ \mc H $. Suppose that $ \mc A $ is spanned by an orthonormal basis  $ w_1, \ldots, w_{k} $  and  $ \mc B$ is spanned by an orthonormal basis $ \tilde w_1, \ldots, \tilde w_{l} $.  Let $ M $ be the discriminant matrix of  $ \lr{ \lrb{w_1, \ldots, w_k}, \lrb{\tilde w_1, \ldots, \tilde w_l}  } $.
The spectral theorem from \cite{Szegedy2004b}  provides spectral decomposition for the operator  $ \refB \refA $ restricted to $ \mc A + \mc B = (\mc A^\perp  \cap \mc B^\perp  )^\perp  $; 
on the subspace $ \mc A^\perp  \cap \mc B^\perp  $   (called the \textit{idle subspace} in \cite{Szegedy2004b}) $ \refB \refA $ acts as the identity.
\begin{theorem}[{\cite[Spectral Theorem]{Szegedy2004b}}]  \label{th:spectralThm}

             $ \refB \refA $ has the following eigenvalues on $ \mc A + \mc B $:
	\begin{itemize}
		\item eigenvalue 1, with eigenvectors being all vectors in $ \mc  A  \cap \mc  B    $; the space has the same dimension as the eigenspace of $ M $ associated to the eigenvalue 1.
		\item eigenvalue $ 2\lambda^2 - 1 \mp 2i \lambda \sqrt{1 - \lambda^2} $, with corresponding eigenvectors  $ \lr{\ket{\tilde a} - \lambda \ket{\tilde b}}  \pm i \sqrt{1 - \lambda^2} \ket{\tilde b}$, where $ (a,b) $ is  an eigenvector of $ M $ with eigenvalue $ \lambda \in (0,1) $ and 
		\[ 
		\ket{\tilde a} :=\sum_{i=1}^{{k}} a_i w_i, \quad  \ket{\tilde b} := \sum_{j=1}^{{l}} b_j \tilde w_j.
		\]
		\item eigenvalue $ -1 $, with eigenvectors being all vectors in form $  \ket{\tilde a}$ or $ \ket{\tilde b} $, where $ (a,b) $ is an eigenvector of $ M $ with eigenvalue 0.
	\end{itemize}
\end{theorem}

Since
\[ 
\refX{\mc A^\perp}=-\refA, \quad 
\refX{\mc B^\perp}=-\refB, \qquad
\refX{\mc B^\perp}\refX{\mc A^\perp} = \refB \refA ,
\]
we can restrict the operator $ \refX{\mc B^\perp}\refX{\mc A^\perp} $ on $ \mc A + \mc B $ and obtain its spectral decomposition in terms of  the discriminant matrix of the pair $ \lr{ \lrb {w_1, \ldots, w_k}, \lrb {\tilde w_1, \ldots, \tilde w_l}} $ (instead of forming the discriminant matrix of the orthogonal systems spanning  $ \mc A^\perp $ and $ \mc B^\perp  $).

\subsection{1-eigenvectors of  \texorpdfstring{$\RB \RA$}{RB RA}}
\label{sec:1eigen}

In our setting  $ \RB = \refX{\spaceB^\perp} $ and $ \RA = \refX{\spaceA ^\perp} $.

From the spectral theorem it follows that all 1-eigenvectors of $ \RB \RA $   belong to either  $ \spaceA ^\perp \cap  \spaceB^\perp  $ or $  \spaceA   \cap  \spaceB $. 
We start with  characterizing these 1-eigenspaces  as follows:
\begin{lemma}\label{th:1eigenvectors}
	\begin{enumerate}
		\item The starting state $ \ket {\addE} $ is orthogonal to each state in $ \spaceA ^\perp \cap  \spaceB^\perp  $.
		\item $ \dim \lr{ \spaceA   \cap  \spaceB}  =0$.
	\end{enumerate}
\end{lemma}
\begin{proof}
	The first claim can be restated as
	\[ 
	 \ket {\addE} \in \lr{\spaceA ^\perp \cap  \spaceB^\perp }^\perp  = \spaceA +\spaceB .
	 \]
	To show that, notice that the following equality holds:
	\begin{equation}\label{eq:l13e01}
	\sum_{v \in \vertSetA} \sum_{e \in \nhood(v)} \ket e 
	=
	\sum_{v \in \vertSetB} \sum_{e \in \nhood(v)} \ket e ,
	\end{equation}
	since for every  edge $ e \in \edgeSet $ the state $ \ket e $ is added both in the LHS and RHS of \eqref{eq:l13e01} exactly once: if $ e  $ connects a vertex $ u \in \vertSetA  $ and a vertex $ v \in \vertSetB $ (no edge can connect two vertices from $ \vertSetA $ or from $ \vertSetB $), then $ \ket e $ appears in the sum $ \sum_{e' \in \nhood(u)} \ket {e'} $ in the LHS and  in the sum $ \sum_{e' \in \nhood(v)} \ket {e'} $ in the RHS (and, for any other vertex $ w \notin \lrb{u,v} $, $ \ket e $ is not contained in the sum $ \sum_{e' \in \nhood(w)} \ket {e'} $).
	
	Now from \eqref{eq:l13e01} we conclude that 
	\[ 
	\ket \addE = \ket \addE  + \alpha \sum_{v \in \vertSetA} \sum_{e \in \nhood(v)} \ket e  - \alpha \sum_{v \in \vertSetB} \sum_{e \in \nhood(v)} \ket e 
	=
	 \ket {s_{v_1}} +  \alpha \sum_{v \in \vertSetA \setminus \lrb{v_1}} \ket {s_v} -  \alpha \sum_{v \in \vertSetB} \ket {s_v},
	 \]
	i.e., $\ket {\addE} \in  \spaceA +\spaceB   $ as claimed.
	
	To show the second  claim, suppose  a vector $ \ket x $ is contained both in $ \spaceA $ and $  \spaceB$.  Then $ \ket x $ can be expressed in two ways via the  vectors $ \ket {s_v} $, i.e., there are scalars $ \eta_i $, $ i \in [V] $, such that 
	\begin{equation}\label{eq:l13e02} 
	\sum_{i=A+1}^{A+B} \eta_i \sum_{e \in \nhood(v_i)}  \ket e
	=
	\sum_{i =1}^A  \eta_i \sum_{e \in \nhood(v_i)}  \ket e + \eta_1 \alpha^{-1} \ket \addE.
	\end{equation}
	
	Clearly, for any adjacent vertices $ v_i \sim v_j $ we must have $ \eta_i = \eta_j $, since the corresponding basis state $ \ket e $ (where $ e  $ is the unique edge between $ v_i $ and $ v_j $) has coefficient $ \eta_i $ in the LHS of \eqref{eq:l13e02} (supposing that $ v_i $ is in $ \vertSetA $) and coefficient $ \eta_j $ in the RHS of \eqref{eq:l13e02}. However, $ \mc G $ is connected, therefore we must have $ \eta_1 = \eta_2 = \ldots = \eta_V $. It remains to notice that  $ \eta_1 = 0 $, since the LHS of \eqref{eq:l13e02} is orthogonal to $ \ket \addE  $. 
	We conclude that  only the null vector belongs to the subspace $ \spaceA   \cap  \spaceB $.
\end{proof}
An immediate consequence of this Lemma is that    all   1-eigenvectors  $\ket{\psi}$ of $\RB \RA$  are orthogonal to the starting state $ \ket \addE$.

\subsection{Eigenvalue closest to 1}\label{sec:lem13}
The spectral decomposition for $ \RB \RA $ (restricted to $\spaceA +\spaceB  $) will be  obtained from the discriminant matrix of  two orthonormal systems spanning $ \spaceA  $ and $ \spaceB   $; in $ \lr{\spaceA +\spaceB }^\perp  = \spaceA ^\perp \cap  \spaceB^\perp $ the operator $ \RB \RA $ acts as the identity.

From     Theorem \ref{th:spectralThm} and Lemma \ref{th:1eigenvectors} it follows that the discriminant matrix does not have the eigenvalue 1.
Let $ \lambda \in (0,1) $ be the maximal eigenvalue of the discriminant matrix and  $\theta = 2 \arccos \lambda$, then $ e^{\pm i \theta} $ is the  eigenvalue of $ \RB \RA $ which is closest to 1.

To describe the discriminant matrix, we introduce the following notation.
We define a $ (T+1) \times A $ matrix $ \mata $ as follows:  
\begin{itemize}
	\item the elements of the first column are defined by
	\[ 
	\mata[i,1] = 
	\begin{cases}
	1, & e_i \in \nhood(v_1),\\
	\alpha^{-1}, & i=T+1, \\
	0, & \text{otherwise};
	\end{cases}
	\]
	\item the elements of the $ j^\text{th} $ column, $ j=2,3,\ldots,A $, are defined by
	\[ 
	\mata[i,j] = 
	\begin{cases}
	1, &  e_i \in \nhood(v_j), \\
	0, & \text{otherwise.}
	\end{cases}
	\]
\end{itemize}
A $ (T+1) \times B $ matrix $ \matb $ is defined by
\[ 
\matb[i,j]  = 
\begin{cases}
1, & e_i \in \nhood(v_{A+j}), \\
0, & \text{otherwise,}
\end{cases}
\quad i \in [T+1], \ j \in [B].
\]
Then $ \spaceA $ and $ \spaceB $ can be identified with the  column spaces of   $ \mata $ and $ \matb $, respectively.  
Let $ a \in \mbb R^{A} $ and and  $ b \in \mbb R^{B} $   be  vectors defined by
\begin{equation}\label{eq:vec_ab_def}
a[1] = \sqrt{d_1 +  \alpha^{-2} }, \   a[j] = \sqrt{d_j}, \ j \in [2 \tdots A]
\quad \text{ and } \quad
b[j] = \sqrt{d_{A+j}}, \ j \in [B].
\end{equation}
By $ \matA $ we denote the matrix $ \mata  \diag(a)^{-1} $ and by $ \matB $ we denote the matrix $ \matb  \diag(b)^{-1} $. Notice that columns of $ \matA $ and $ \matB $ are orthonormal      vectors. The corresponding vectors
\begin{equation}\label{eq:HAHB_basis}
\sum_{i \in [T+1]} \matA[i,j] \ket i =   \frac{\ket {s_{v_j}}}{\nrm{s_{v_j}}}, \ j \in [A],
\quad  \text{ and  } \quad
\sum_{i \in [T+1]} \matB[i,j] \ket i  =   \frac{\ket {s_{v_{A+j}}}}{\nrm{s_{v_{A+j}}}}, \ j \in  [B]
\end{equation}
form orthonormal  bases of $ \spaceA $ and $ \spaceB $, respectively.

Let $ L = \matA^* \matB $. Then the discriminant matrix   of the pair of orthonormal systems spanning $ \spaceA $ and $ \spaceB $ has the following block structure:
\[ 
\begin{pmatrix}
\mb 0_{A,A} &  L \\
L^* & \mb 0_{B,B} 
\end{pmatrix}.
\]

It is easy to check that $ \lr{\begin{smallmatrix} u \\ v \end{smallmatrix} }$ is  an eigenvector of the discriminant matrix with an eigenvalue $ \lambda >0$ iff  $ u $ is  a left-singular  vector and $ v $
is the right-singular vector of $ L $ with singular value $ \lambda $.  Therefore, 
if $ \lmax{L} = \cos \frac{\theta}{2} $ is the largest singular value of $ L $,  
then $ e^{\pm i \theta} $ are the eigenvalues of $ \RB \RA $ that are closest to 1.

Let $ K=\lr{I - LL^*}^{-1}  $
and $ \lRho $ be  the maximal eigenvalue of $ K $. Since $  \lmax{L}   =  \cos \frac{\theta}{2} < 1  $  is  the maximal singular value of $ L $, it holds that $\lRho = \lr{1-  \lmax{L} ^2}^{-1}  =  \sin^{-2} \frac{\theta}{2}  $ is  the maximal eigenvalue of $ K$.
We show the following characterization of the entries of $ K $:
\begin{lemma}\label{lem:corr9}
	For all $ i,j  \in  [A]$    the following inequalities hold:
	\[ 
	\alpha^2     a[i] a[j]
	\leq
	K[i,j] \leq  
	\lr{\alpha^2    + n} a[i] a[j]
	.
	\]
	Moreover, when $ i=1 $ or $ j=1 $, we have $ K[i,j] = {\alpha^2    } a[i] a[j] $.
\end{lemma}
Proof in Sections \ref{sec:harmonic} - \ref{sec:appE}.

Lemma \ref{lem:corr9} now allows to estimate   the maximal eigenvalue of $ K$:
\begin{lemma}\label{th:lemma2}
The entries $K[i, j], i,j  \in  [A]$, satisfy
	\begin{equation}\label{eq:lemma14_eq1}
	\alpha^2   \sqrt{d_i d_j}
	\leq
	K[i,j] \leq  
	\sqrt{d_i d_j}\lr{\alpha^2    + n} .
	\end{equation}
	Furthermore, $ \lRho $, the maximal eigenvalue of $ K$,  satisfies
	\begin{equation}\label{eq:lemma14_eq2}
	\alpha^2 T \leq \lRho \leq (\alpha^2 +  n) T.
	\end{equation}
\end{lemma}
\begin{proof}
	Since $ a[i] = \sqrt{d_i} $ for all $ i \in [2 \tdots A] $, inequalities \eqref{eq:lemma14_eq1} immediately follow from Lemma \ref{lem:corr9} when $ i,j \geq 2 $. Suppose that $ i=1 $ (since $ K $ is symmetric), then, again by  Lemma \ref{lem:corr9}, $ K[1,j] = \alpha^2 a[1] a[j] $. Since $ a[j] \geq \sqrt{d_j} $ for all $ j \in [A] $, the first inequality in \eqref{eq:lemma14_eq1} is obvious. It remains to show
	\begin{equation}\label{eq:lemma14_eq3}
	\alpha^2 a[1] a[j]  \leq \sqrt{d_1 d_j}\lr{\alpha^2    + n}, \quad j \in [A].
	\end{equation}
	 However, from the definition of $ a $  we have $ \alpha a[j] \leq \sqrt{\alpha^2 d_j + 1} $ for all $ j \in [A] $. Hence the LHS of \eqref{eq:lemma14_eq3} is upper-bounded by $ \sqrt{ (\alpha^2 d_1+ 1)(\alpha^2 d_j+ 1)  } $, which, in turn, is upper bounded by the RHS of \eqref{eq:lemma14_eq3}. This proves  \eqref{eq:lemma14_eq1}.

	Let $ K' $ be a symmetric $ A \times A $ matrix, defined by $ K'[i,j] = \sqrt{d_i d_j} $.
	Then  \eqref{eq:lemma14_eq1} can be restated as
	\begin{equation*}  
	\alpha^2 K '[i,j] \leq K[i,j] \leq (\alpha^2 + n)K' [i,j], \qquad i,j \in [A].
	\end{equation*}
	Now, from  \cite[Theorem 8.1.18]{Horn2012b}   we have that 
	\[ 
	\lambda(\alpha^2 K') \leq \lRho  \leq \lambda\lr{(\alpha^2 + n)K' },
	\]
	where by $ \lambda(M) $ we denote the spectral radius of a matrix $ M $ (the maximum absolute value of an eigenvalue of $ M $),  since $ \lRho = \lambda(K) $. On the other hand, $K'$ is a rank-1 matrix, thus its spectral radius is $ \lambda(K') =  \sum_{j=1}^{A} (\sqrt{d_j} )^2 = \sum_{j=1}^{A} d_j = T$ (on each side of the last equality  every edge {in $\edgeSet$} is counted exactly once),  hence
	\[ 
	\alpha^2 T= \lambda(\alpha^2K') 
	\leq \lRho  \leq 
	\lambda\lr{(\alpha^2 + n)K' } = (\alpha^2 +n)T .
	\]
\end{proof}

\subsection{Phase estimation for \texorpdfstring{$\theta$}{theta}}
\label{sec:phase}

We now show that Lemma \ref{th:lemma2} implies Lemmas \ref{lem:prob-success} and \ref{lem:theta}.

Let $ u $ be  the left-singular  vector and $ v $ be the right-singular vector of $ L $ corresponding to the largest singular value $  \lmax{L}  $. By Theorem \ref{th:spectralThm}, the corresponding eigenvectors of $ \RB \RA $ are
\[ 
\ket {\Psi_\pm}  = {\frac{1}{ \sqrt{2 (1- \lmax{L} ^2)}}\ket{\tilde a} - \frac{ \lmax{L} }{ \sqrt{2 (1- \lmax{L} ^2)}} \ket{\tilde b}}  \mp \frac{i}{\sqrt 2} \ket{\tilde b},
\] 
where $\ket{\tilde a} \in \spaceA$, $ \ket{\tilde b} \in \spaceB$ are the unit vectors  associated to $   \matA  {u} $ and $   \matB  {v} $, i.e.,
\[ 
\ket{\tilde a} =   \sum_{j=1}^{A} \frac{\ket {s_{v_j}}}{\nrm{s_{v_j}}} u[j],
\quad
\ket{\tilde b} =   \sum_{j=1}^{B} \frac{\ket {s_{v_{A+j}}}}{\nrm{s_{v_{A+j}}}} v[j].
 \]
 
The two-dimensional plane  $ \Pi = \Span \lrb{\ket {\Psi_+}, \ket {\Psi_- }} $ is also spanned by 
\begin{equation}\label{eq:def_q1q2}
\ket {q_1} = \ket{\tilde b}, \quad
\ket {q_2} =   \frac{1}{\sqrt{1 -  \lmax{L} ^2}} \ket{\tilde a}  - \frac{ \lmax{L} }{\sqrt{1 -  \lmax{L} ^2}} \ket{\tilde b}.
\end{equation}
We claim that 
\begingroup
\def\thelemmaPrime{\ref*{lem:prob-success}}
\begin{lemmaPrime}
	If $ \alpha\geq \sqrt {2n} $, we have
	\[ \braket{\addE}{q_2}   \geq   \frac{2}{3} \]
	for the state $\ket{q_2}\in  \Span \lrb{\ket {\Psi_+}, \ket {\Psi_-}} $,  defined by \eqref{eq:def_q1q2}.
\end{lemmaPrime}
\endgroup
\begin{proof} 
Since $ \ket {\addE} \perp \spaceB $,  we have $ \braket{\addE}{\tilde b} = 0 $ and 
\[ 
 \braket{\addE}{q_2}
=
\frac{1}{\sqrt{1- \lambda_L^2}}
\braket{\addE}{\tilde a} =\sqrt{ \lRho } \braket{\addE}{\tilde a}.
 \]
 Since $ \braket{\addE} {s_{v_j}} = 0$  unless $ j=1 $, we have
 \[ 
 \braket{\addE}{\tilde a}
 =
 \sum_{j=1}^{A} \frac{\braket{\addE} {s_{v_j}}}{\nrm{s_{v_j}}} u[j]
 =
 \frac{ u[1]  }{ \sqrt{ 1 + \alpha^2 d_1} }.
 \]
Consequently,
\begin{equation}\label{eq:braket_1q2}
\braket{\addE}{q_2} = \frac{u [1]\sqrt{\lRho }}{\sqrt{ 1 + \alpha^2 d_1}}.
\end{equation}
	We now lower bound this expression.
          Since $u$ is an eigenvector of $K$, we have
	\[ 
	\lRho u[i] =  \sum_{j=1}^{A}   u[j] K[i,j]  , \quad i \in [A] .
	\]
	From \eqref{eq:lemma14_eq1} 
	it follows that
	\[ 
	 \lRho u[i] 
	\leq ( \alpha^2 +n)\sqrt{d_i } \lr{\sum_{j=1}^{A} \sqrt{d_j} u[j]}, \quad i \in [A] .
	 \]
	 Denote $ \mu =  \sum_{j=1}^{A} \sqrt{d_j} u[j]  $; then the previous inequality can be rewritten as
	 \[ 
	 u[i] 
	 \leq \mu \cdot \frac{ \alpha^2 +n}{\lRho}\sqrt{d_i } , \quad i \in [A].
	 \]
	 On the other hand, $ u $ is a unit vector, thus  
	 \[ 
	 1 = \sum_{i=1}^{A} u^2[i]   \leq \mu^2 \lr{ \frac{\alpha^2+n}{\lRho} }^2 \sum_{i=1}^{A}d_i =  \mu^2 \lr{ \frac{\alpha^2+n}{\lRho} }^2 T.
	 \] 
	 Now we conclude that
	 \[ 
	 \mu \geq   \frac{\lRho}{(\alpha^2+n) \sqrt{T}} .
	 \]
	 
	 From Lemma \ref{lem:corr9} it follows that $ K[1,j] = \alpha^2 a[1] a[j] \geq     \alpha^2 \sqrt{d_1 + \alpha^{-2}} \sqrt{d_j}$ for all $ j \in [A] $. That allows to estimate the RHS of the equation
	 \[ 
	  u[1] = \frac{1}{\lRho}  \sum_{j=1}^{A}   u[j] K[1,j]  
	  \]
	 more precisely:
	 \[ 
	 u[1] \geq  \frac{ \alpha^2 \sqrt{d_1 + \alpha^{-2}}}{\lRho}   \sum_{j=1}^{A}   \sqrt{d_j} u[j] 
	 =
	 \frac{ \alpha^2 \sqrt{d_1 + \alpha^{-2}}}{\lRho} \mu \geq \frac{\alpha^2}{\alpha^2+n} \cdot \frac{\sqrt{\alpha^2 d_1 + 1} }{\alpha \sqrt{T}}.
	 \]
	 Combining this with \eqref{eq:braket_1q2} and  the estimate $ \sqrt{\lRho} \geq \alpha\sqrt{T} $ (which follows from  \eqref{eq:lemma14_eq2}) yields
	 \[ 
	 \braket{\addE}{q_2} =
	 \frac{ u[1] \sqrt{\lambda_K}}{ \sqrt{ 1 + \alpha^2 d_1} } \geq  \frac{\alpha^2}{\alpha^2+n} = 1 - \frac{n}{\alpha^2+n} \geq \frac{2}{3},
	 \] 
	which completes the proof.
\end{proof}

\begingroup
\def\thelemmaPrime{\ref*{lem:theta}}
\begin{lemmaPrime}
	Suppose that  $ \delta \in (0,1) $. Let    $ \alpha = \sqrt{2 n  \delta^{-1}} $   and $ \hat \theta \in  (0;\pi/2)$ satisfy
		\[ 
		\lrv{\hat \theta - \theta } \leq \frac{\delta^{1.5}}{ {24} \sqrt { 3nT }}.
		\] 
	Then
	\[ 
	(1 - \delta) T \leq   \frac{1}{\alpha^2 \sin^2 \frac{\hat \theta }{2}}  \leq  (1+ \delta)T.
	\]
\end{lemmaPrime}
\endgroup

\begin{proof}
	We start with
	
	\begin{lemma}\label{th:lemma14}
		Suppose that   $ \hat \theta \in (0;\pi/2) $ satisfies
		\[ 
		\lrv{\theta - \hat \theta}  \leq \epsilon \sin \frac{\theta }{2} 
		\]
		for some {$ \epsilon \in (0,1) $}.  Then
		\[ 
		\lrv{\frac{1}{\sin^2 \frac{\hat \theta }{2}} - \frac{1}{\sin^2 \frac{\theta }{2}} } \leq 
		\frac{  {6}\epsilon}{   \sin^2\frac{\theta}{2}   }.
		\]
	\end{lemma}
	
	\begin{proof}
		Let $ \sigma  = \sin \frac{\theta}{2}$, then $    \lrv{\theta - \hat \theta}  \leq \epsilon  \sigma$ and we have to show that
		\[ 
		\lrv{\frac{1}{\sin^2 \frac{\hat \theta }{2}} - \frac{1}{\sin^2 \frac{\theta }{2}} } \leq 
		\frac{ {6}\epsilon}{\sigma^2}.
		\]
		
		Since the $\sin$ function is Lipschitz continuous with Lipschitz constant 1,  we have
		\[ 
		\lrv{\sin \frac{\hat \theta}{2} - \sin \frac{\theta}{2}} \leq   \frac{\epsilon \sigma}{2}
		\]
		and
		$ {\ensuremath{ \sigma \lr{ 1- \frac{\epsilon}{2} }  \leq  }} \sin \frac{\hat \theta}{2} \leq  \sigma \lr{ 1+ \frac{\epsilon}{2} }$.
		On the other hand,
		\[ 
		\lrv{\sin^2 \frac{\hat \theta}{2} - \sin^2 \frac{\theta}{2}} 
		=
		\lrv{\sin \frac{\hat \theta}{2} - \sin \frac{\theta}{2}  }
		\lrv{\sin \frac{\hat \theta}{2} + \sin \frac{\theta}{2}  }
		\leq 
		\frac{\epsilon \sigma}{2} \cdot 
		2 \max \lrb{\sin \frac{\hat \theta}{2} , \sin \frac{\theta}{2}   }
		\leq \epsilon  \lr{1+ \frac{\epsilon}{2}}\sigma^2.
		\]
		We conclude that 
		\[ 
		\lrv{\frac{1}{\sin^2 \frac{\hat \theta }{2}} - \frac{1}{\sin^2 \frac{\theta }{2}} }
		=
		\frac{\lrv{\sin^2 \frac{\hat \theta}{2} - \sin^2 \frac{\theta}{2}} }{ \sin^2 \frac{\hat \theta}{2} \cdot \sin^2 \frac{\theta}{2} }
		\leq
		\frac{ \epsilon  \lr{1+ \frac{\epsilon}{2}}\sigma^2}{\sigma^4 \lr{1   -\epsilon + \epsilon^2/4}}
		\leq
		\frac{\epsilon (1 + \epsilon/2)}{\sigma^2 \lr{1 -\epsilon + \epsilon^2/4}}
		< 
		\frac{6\epsilon}{\sigma^2},
		\]
  {where  the last inequality is due to $2x(1+x)/(1-x)^2 < 12x$, which is valid for all $x \in (0,1/2)$, and, in particular, for $x=\epsilon/2 $.
  }
	\end{proof}

	From Lemma \ref{th:lemma2} it follows that $ \sigma :=   \sin \frac{\theta}{2}$  satisfies
	\begin{equation}\label{eq:cor1_eq1}
	\alpha^2 T \leq    \sigma^{-2} \leq    (\alpha^2 + n) T.
	\end{equation}
	Notice that
	\[ 
	\lrv{\theta - \hat \theta}
	\leq \frac{\delta^{1.5}}{{24} \sqrt { 3nT }} 
	=
	\frac{\delta}{{24} } \cdot \sqrt{\frac{\delta}{3}}     \cdot \frac{1}{\sqrt{n T}}
	<
	\frac{\delta}{{24} } \cdot \sqrt{\frac{\delta}{2 + \delta}}     \cdot \frac{1}{\sqrt{n T}}
	=
	\frac{\delta}{{24}  \sqrt{ \lr { \frac{2n}{\delta} + n } T} }  .
\]
Since $\frac{2n}{\delta}=\alpha^2$,    the RHS of the last equality is upper-bounded by  \eqref{eq:cor1_eq1} as $  \frac{\delta }{24 \sqrt{( \alpha ^2 + n)T}} \leq   \frac{\delta\sigma}{24}$. An application of Lemma \ref{th:lemma14} (with $\epsilon=\delta/24 < 1$) now gives
	\[ 
	\lrv{\frac{1}{\sin^2 \frac{\hat \theta }{2}} -  \sigma^{-2}} 
	\leq 
	\frac{\delta\sigma^{-2}}{4}   
	\]
	or
	\[ 
	\sigma^{-2} \lr{1  - \frac{\delta}{4}}
	\leq
	\frac{1}{\sin^2 \frac{\hat \theta }{2}} 
	\leq 
	\sigma^{-2} \lr{1+ \frac{\delta}{4}}.
	\]
	From    \eqref{eq:cor1_eq1}  it follows that
	\[ 
	\lr{1  - \frac{\delta}{4}} \alpha^2 T
	\leq  
	\frac{1}{\sin^2 \frac{\hat \theta }{2}} 
	\leq 
	\lr{1+ \frac{\delta}{4}} (\alpha^2 + n) T.
	\]
	Consequently, 
	\[ 
	\lr{1  - \frac{\delta}{4}}   T
	\leq  
	\frac{1}{\alpha^2 \sin^2 \frac{\hat \theta }{2}} 
	\leq 
	\lr{1+ \frac{\delta}{4}}  \lr{1  + \frac{n}{\alpha^2}}  T
	=
	\lr{1+ \frac{\delta}{4}}  \lr{1  + \frac{\delta}{2}}  T.
	\]
	It remains to notice that
	\[ 
	\lr{1+ \frac{\delta}{4}}  \lr{1  + \frac{\delta}{2}}
	=
	1 + \frac{3}{4} \delta + \frac{\delta^2}{8} < 1 + \delta.
	\]
	Hence
	\[ 
	\lr{1  - \frac{\delta}{4}}   T
	\leq  
	\frac{1}{\alpha^2 \sin^2 \frac{\hat \theta }{2}} 
	\leq 
	\lr{1 +  \delta }    T,
	\]
	and the claim follows.
	
\end{proof}

\subsection{Harmonic functions and electric networks}
\label{sec:harmonic}

The next five subsections are devoted to the proof of Lemma \ref{th:lemma2}.
We start with the concept of harmonic functions which is linked to the connection between electric networks and random walks.  More details on the subject can be found in \cite[Sec. 4]{Lovasz1993} (in the case of a simple random walk which is the framework we use below), \cite[Lect. 9]{Lawler1999},  \cite[Chap. 2]{LP:book}  and  \cite[Chap. 9]{Levin2009} (in a more general setting with weighted graphs).

Throughout the rest of this subsection suppose that $ \mb P $ is a transition matrix of an irreducible Markov chain with a state space $ \Omega $.

\begin{definition}
	Consider a function $ f : \Omega  \to \mbb R $. We say that  $ f $ is \textit{harmonic} for $ \mb P $ at $ x \in \Omega $ (or simply ``harmonic at  $ x $'') if 
	\cite[Eq. 1.28]{Levin2009}
	\[ 
	f(x) = \sum_{ y \in \Omega}  \mb P[x,y] f(y) ,
	\]
	i.e., $ f  $ has the \textit{averaging property} at $ x $.
	If $ f $ is harmonic for every $ x \in  \Omega' \subset \Omega$, then  $ f $ is said to be harmonic on the subset $ \Omega' $.
\end{definition}

It can be easily seen that a linear combination of harmonic functions is still harmonic. In particular, all constant functions are harmonic on $ \Omega $.

  It is known that a harmonic (on $\Omega' $) function attains its maximal and minimal values (in the set $\Omega  $) on     $ \Omega \setminus  \Omega' $ (it appears as an exercise in \cite[Exercise 1.12]{Levin2009}; in the context of weighted random walks it can be found in \cite[Lemma 3.1]{Pemantle1995}).
\begin{lemma}[{Maximum Principle}]\label{th:maxPr}
	Let $ f $ be harmonic for $ \mb P $ on $ \Omega' \subset \Omega $, with $  \Omega \setminus  \Omega' \neq \emptyset $. Then there exists $ x \in  \Omega \setminus \Omega' $ s.t. $ f(x) \geq f(y) $ for all $ y \in \Omega' $.
\end{lemma}
By applying the Maximum Principle for $ -f $, one obtains a similar statement for the minimal values of $ f $. In particular, if the function $ f $ is harmonic on $ \Omega' $ and constant on the  set $  \Omega \setminus \Omega' $, then $ f $ is constant. A consequence of this is the Uniqueness Principle:  if $ f $ and $ g  $ are harmonic on $ \Omega'$ and $ f \equiv g  $ on $  \Omega \setminus \Omega' $, then $ f \equiv g $ on $ \Omega  $. Moreover, when $ f $ is  harmonic everywhere on $ \Omega $, it is a constant function.

Further, suppose that $ \mb P $ is a simple random walk on a finite connected  undirected graph $ \mc G =(\Omega,\mc E)$,  in the sense that
\[ 
\mb P[x,y] =   \frac{\mbb 1_{\lrb{x \sim y}}}{d(x)}, \quad d(x) :=  \lrv {\vbl \lrb{y \in \Omega  \ \vline\  x \sim y}} ,
\]
where $ \mbb 1 $ stands for the indicator function.
Then the Markov chain, corresponding to $ \mb P $, is time-reversible; the  graph $ \mc G $ can be viewed as an electric network where each edge has unit conductance
(this approach can be further  generalized to  weighted graphs where the weight  $ c(e) $ of an edge  $ e \in \mc E $  is  referred to as  {conductance} in the electric network theory, whereas its reciprocal  $ r(e)  = c(e)^{-1} $ is called  {resistance}). 

For a subset $ \Omega' \subset \Omega$  we call the   set $ \lrb{ x \in  \Omega \setminus \Omega'   \ \vline\   \exists y \in \Omega' :  x \sim y} $  the boundary of $ \Omega' $ and denote  by $ \partial \Omega' $.  
 The Maximum Principle can be strengthened as follows:
\begingroup
\def\thelemmaPrime{\ref*{th:maxPr}'}
\begin{lemmaPrime}
	Let $ f $ be harmonic for $ \mb P $ on $ \Omega' \subset \Omega $, with $ \partial \Omega' \neq \emptyset $. Then there exists $ x \in \partial \Omega' $ s.t. $ f(x) \geq f(y) $ for all $ y \in \Omega' $.
\end{lemmaPrime}
\endgroup
Similarly, 
\begin{itemize}
	\item  if the function $ f $ is harmonic on $ \Omega' $ and constant on the boundary $ \partial \Omega' $ (for example, when the boundary is  a singleton), then $ f $ is constant on $ \Omega' \cup \partial \Omega' $. 
	\item if $ f $ and $ g  $ are harmonic on $ \Omega'$ and $ f \equiv g  $ on $ \partial \Omega' $, then $ f \equiv g $ on $ \Omega' \cup \partial \Omega' $.
\end{itemize}

Let $ s, t \in \Omega $ be two different vertices of the graph $ \mc G $ and consider the unique function $ {\phi_{st}  : \Omega \to \mbb R} $, which 
\begin{itemize}
	\item is harmonic on $ \Omega \setminus \lrb{s,t} $,
	\item satisfies boundary conditions $ \phi_{st}(s)  = 1 $, $ \phi_{st}(t) = 0 $.
\end{itemize}
From the Maximum Principle it follows that values of $ \phi_{st} $ are between 0 and 1.
In the electric network theory, the function $ \phi_{st}(u) $, $ u\in\Omega $, is interpreted as the voltage of $ u $, if we put current through the electric network associated to $ \mc G $, where the voltage of $ s $ is 1  and the voltage of $ t $ is 0 (see \cite[p. 22]{Lovasz1993} or \cite[p. 60]{Lawler1999}; in the latter $ \phi_{st} $ is denoted by $ \tilde V $). For the random walk with the transition matrix $ \mb P $, the value $ \phi_{st}(u) $ is the probability that  the random walk  starting at $ u$ visits $ s $ before $ t $.

Consider the quantity
\[ 
R_{st}  =    \lr{\sum_{u: u \sim t}  \phi_{st}(u)   }^{-1}.
\]
The electrical connection is that this quantity, called the \textit{effective resistance} (or simply resistance  between $ s $ and $ t $ in \cite{Lovasz1993}), is the voltage difference that would result between $ s $ and $ t $  if one unit of current was driven between them. On the other hand, it is linked to the ``escape probabilities'', because $   \lr{d(t) R_{st}}^{-1}  $ is the probability that a random walk, starting at $ t $, visits $ s $ before  returning back to $ t $ \cite[p. 61]{Lawler1999}.

An important result is the Rayleigh's Monotonicity Law \cite[Theorem 9.12]{Levin2009}, from which it follows that adding edges in the graph cannot increase the effective resistance \cite[Corollary 4.3]{Lovasz1993}, \cite[Corollary 9.13]{Levin2009}. More precisely, if we add another edge in $ \mc G $, obtaining a graph $ \mc G'  = (\Omega, \mc E') $, and denote by $ R_{st}' $ the effective resistance between vertices $ s $ and $ t $, then $ R_{st} \geq R_{st}' $. 

More generally, if we consider the same graph $ \mc G $, but with different weights (or conductances)  $ c(x,y) $ and $ c'(x,y) $, satisfying $ c(x,y) \leq c'(x,y) $ for all $ x,y \in \Omega $,
then the Monotonicity Law says that the effective resistances satisfy  the opposite inequality $ R_{st} \geq R_{st}' $ for all distinct $ s,t \in \Omega $.

We can view adding an edge as increasing its weight from 0 to 1, hence the claim about edge adding.

\subsection{Extended DAG and an absorbing random walk}

Let $ \mc G $ and $ \mc G' $  be as defined in Section \ref{sec:GSE}.
In this Section, we define an absorbing random walk on $\mc G''$, a slightly extended version of $\mc G'$.

Let $ \Gamma $  denote the adjacency matrix of the weighted graph   $ \G' $.
We denote $ \beta = \lr{d_1 \alpha^2 + 1}^{-1} $.

We introduce another  vertex  $\addVV $ and connect  the vertex $ \addV $ with an edge $ \addEE = (\addVV, \addV) $. Let  $ \complVV = \complV \cup \lrb{\addVV} $, $ \complEE= \complE \cup \lrb{\addEE} $ and $ \complGG = \lr{\complVV,\complEE} $.

Finally, let all  edges in the original DAG $ \G $ have weight 1, but the two additional  edges  have the following weights:
\begin{itemize}
	\item the edge $ \addE $ has weight $ \alpha^{-2}  $; 
	\item the edge $ \addEE $ has weight $ d_1  $.
\end{itemize}

For any two vertices $ v_i ,v_j$,  $i,j \in  [\noTot+2] $, we denote $ v_i \sim v_j$ if there is an edge between $ v_i $ and $ v_j $ in the DAG $ \G'' $.  For each   $ i\in  [\noTot+2] $  we denote by $ d''(i) $ the degree of the vertex $ v_i $  in the weighted graph $ \complGG $. Then
\begin{itemize}
	\item $ d''(1) = d''(\noTot+1) =d_1 +  \alpha^{-2}  $;
	\item  $ d''(i) $ for each $ i\in [2 \tdots \noTot] $ equals $ d_i $, i.e., the degree of $ v_i $ in  the unweighted graph $ \G $;
	\item $ d''(\noTot+2)  = d_1$.
\end{itemize}
Notice that
\[ 
a[i] = \sqrt{d''(i)},\ i \in [A],
\quad \text{and} \quad
b[j]= \sqrt{d''(A+j)},\ j \in [B],
 \]
for vectors  $ a $ and $ b $  defined with \eqref{eq:vec_ab_def}.

Consider a  random walk on the graph $ \complGG $ with transition probabilities as follows:
\begin{itemize}
	\item when at the vertex $ \addVV $, with probability 1 stay  at $ \addVV $;
	\item when at the vertex $ \addV $, with probability  $ 1-\beta $ move to $ \addVV  $  and with probability $ \beta $ move to $ v_1 $;
	\item when at the vertex $ v_1 $, with probability  $ \beta $ move to $ \addV $  and with probability $ \frac{1-\beta}{d_1} = \beta\alpha^2 $ move to any $  v \in \nhood(v_1)$  (i.e., to any neighbor of the root, different from $ \addV $);
	\item at any vertex $ v_i $, $ i \in [2 \tdots \noTot] $, with probability $ \frac{1}{d''(i)} $ move to any neighbor of $ v_i $.
\end{itemize}
In other words, at any vertex we move to any its neighbors with probability proportional to the weight of the edge, except for the vertex $\addVV$, where we stay  with probability 1. Moreover, this random walk ignores edge direction, i.e., we can go from a vertex of depth $ l $ to a vertex of depth $ l-1 $.

This way an absorbing random walk is defined;  let $ \lrb{Y_k}_{k=0}^\infty $ be the corresponding sequence of random variables, where $ Y_k = j \in [\noTot+2] $ if after $ k $ steps the random walk is  at the vertex $v_j $ (i.e., this sequence is the Markov chain, associated to the absorbing random walk).

Let $ P $ be the transition matrix for this walk; it has the following block structure:
\[ 
P  =   
\begin{pmatrix}
 & & &   &        & & 0 \\
& & & &       & & 0 \\
& &Q & &    &  & 0 \\
& & & &  & & \vdots \\
& & & &   &  &1-\beta \\
0&0 &0 &0 & \ldots  & 0& 1 \\
\end{pmatrix},
\] 
where $ Q $ is a matrix of size $ (\noTot+1) \times (\noTot+1) $ which describes the probability of moving from some transient vertex to another.

Define a $  (\noTot+1) \times  1$ vector $ \vecP $ as follows:
\[ 
\vecP [ i] = a[i] = \sqrt{d''(i)}, \ i \in [A], \quad \vecP[A+j] = b[j] = \sqrt{d''(A+j)},\   j\in [B],  
 \]
 and $ \vecP[V+1]  = \sqrt{d''(V+1)} = \vecP[1] $.
Then  $  \diag(\vecP)^2 Q  $ is the adjacency matrix $ \Gamma $ for the graph $ \complG$.

Let 
\[ 
N  = \mb I_{V+1} + Q  + Q^2 + Q^3 + \ldots =  \lr{\mb I_{\noTot+1} - Q}^{-1}
\]
be the fundamental matrix of this walk. An entry of the fundamental matrix $ N[i,j] $, $ i,j \in [V+1] $, equals the expected expected number of visits to the vertex $ j $ starting from the vertex $ i $, before being absorbed, i.e.,
\[ 
N[i,j] = \mbb E \lrk{ \sum_{k=0}^{\infty} \mbb 1_{\lrb{Y_k=j}} \ \vline\   Y_0 = i   }.
\]

Notice that $ NQ = QN =Q  + Q^2 + Q^3 + \ldots =  N  -\mb  I_{V+1}  $. 
It follows that
\begin{equation}\label{eq:qn=nq}
N[i,j] =
\sum_{l=1}^{V+1} N[i,l] Q[l,j] + \delta_{ij}
=
\sum_{l=1}^{V+1} Q[i,l] N[l,j] + \delta_{ij}
\quad \text{ for all } i,j \in [V+1],
\end{equation}
where by $ \delta_{ij}  $ we denote the Kronecker symbol.

\subsection{Entries of the fundamental matrix}

The purpose of this section is to obtain expressions for entries of the fundamental matrix $N$.
In the next subsection, we will relate those entries to entries of the matrix $K$ from Lemma \ref{th:lemma2}.
This will allow us to complete the proof of Lemma \ref{th:lemma2}.

\begin{lemma}\label{th:L1} 
	\[ 
	N[1,\noTot+1] = N[\noTot+1,1] =N[\noTot+1,\noTot+1] =  \frac{1}{1 - \beta} , \quad N[1,1]  = \frac{1}{\beta(1-\beta)}.
	\]
\end{lemma}
\begin{proof}
	We have
	\[ 
	N[\noTot+1,\noTot+1] = \mbb E \lrk{ \sum_{k=0}^{\infty} \mbb 1_{\lrb{Y_k=\noTot+1}} \ \vline\   Y_0 = \noTot+1}  .
	\]
	Notice that
	\[ 
	\sum_{k=0}^{\infty} \mbb 1_{\lrb{Y_k=\noTot+1}}
	=
	\sum_{k=0}^{\infty} \mbb 1_{\lrb{Y_k=\noTot+1 \text{ and } Y_{k+1} =1}} + 1,
	\]
	since from the vertex $ \addV $ one either moves to $ \addVV $ and gets absorbed or moves to $ v_1 $ and   returns back  to $ \addV $ later. 
	From the linearity of expectation it follows that
	\begin{align*}
		N[\noTot+1,\noTot+1] 
		&=
		1+ \sum_{k=0}^{\infty}\mbb E \lrk{  \mbb 1_{\lrb{Y_k=\noTot+1 \text{ and } Y_{k+1} =1}}   \ \vline\   Y_0 = \noTot+1 }
		\\
		& =1 +  \sum_{k=0}^{\infty} \mbb P \lrk{  Y_k=\noTot+1, Y_{k+1} =1   \ \vline\  Y_0 = \noTot+1  }.
	\end{align*}
	Since
	\begin{align*}
		& \mbb P \lrk{  Y_k=\noTot+1 , Y_{k+1} =1   \ \vline\  Y_0 = \noTot+1   }
		 \\
		& =
		\mbb  P \lrk{  Y_{k+1}=\noTot+1    \ \vline\   Y_{k} =\noTot+1, Y_0 = \noTot+1   }  
		\cdot 
		\mbb  P \lrk{  Y_k=\noTot+1    \ \vline\   Y_0 = \noTot+1} 
		\\
		&=
		\mbb  P \lrk{  Y_{k+1}=\noTot+1    \ \vline\   Y_{k} =\noTot+1  }  
		\cdot
		\mbb  P \lrk{  Y_k=\noTot+1    \ \vline\   Y_0 = \noTot+1 } 
		\\
		&=
		Q[\noTot+1,1]  \, \mbb  P \lrk{  Y_k=\noTot+1\ \vline\   Y_0 = \noTot+1 } ,
	\end{align*}
	we obtain
	\[ 
	N[\noTot+1,\noTot+1] =1+  \beta \sum_{k=0}^{\infty} \mbb P \lrk{  Y_k=\noTot+1    \ \vline\  Y_0 = \noTot+1 }.
	\]
	On the other hand,
	\[ 
	N[\noTot+1,\noTot+1] = \mbb E \lrk{ \sum_{k=0}^{\infty} \mbb 1_{\lrb{Y_k=\noTot+1}} \ \vline\   Y_0 = \noTot+1}  
	=
	\sum_{k=0}^{\infty} \mbb P \lrk{  Y_k=\noTot+1    \ \vline\  Y_0 = \noTot+1},
	\]
	hence
	\[ 
	N[\noTot+1,\noTot+1] = 1 + \beta N[\noTot+1,\noTot+1],
	\]
	from which the equality $ N[\noTot+1,\noTot+1] = (1-\beta)^{-1} $ follows.

	From \eqref{eq:qn=nq} it follows that
	\[ 
	N[\noTot+1,\noTot+1] =\sum_{l=1}^{\noTot+1} Q[\noTot+1,l] N[l,\noTot+1]  + 1 =\sum_{l=1}^{\noTot+1} N[\noTot+1,l] Q[l,\noTot+1]  + 1 .
	\]
	Since $ Q[l,\noTot+1]  $ and $ Q[\noTot+1,l] $ is nonzero only for $ l=1 $, we have
	\[ 
	N[\noTot+1,\noTot+1] = \beta N[1,\noTot+1] +1 =  \beta N[\noTot+1,1] +1 .
	\]
	From that we conclude 
	\[ 
	N[\noTot+1,1]  = N[1,\noTot+1] = \frac{1}{\beta} \lr{ N[\noTot+1,\noTot+1] -1 } = \frac{1}{1-\beta}.
	\]
	
	Finally, again from \eqref{eq:qn=nq} we obtain
	\[ 
	N[\noTot+1,1] =\sum_{l=1}^{\noTot+1} Q[\noTot+1,l] N[l,1]  = \beta N[1,1],
	\]
	thus $ N[1,1] = \lr{ \beta(1-\beta) }^{-1} $.
\end{proof}

Define the matrix $ \tilde N  = N \diag(\vecP)^{-2} $, then 
\[ 
\tilde N[i,j] =  \frac{N[i,j]}{d''(j)}, \quad  i, j \in [\noTot+1].
\]
We note that $ \tilde N $ is a symmetric matrix, since
\[ 
\tilde N =  \lr{\mb I_{\noTot+1} - Q}^{-1} \diag(\vecP)^{-2}   =  \lr{\diag(\vecP)^2 - \diag(\vecP)^2Q   }^{-1}  
\]
and $  \diag(\vecP)^2Q $ is symmetric.  Moreover, from the symmetry  we also have   $ d''(l) Q[l,j]  =   d''(j)Q[j,l] $ for   $ j,l \in [\noTot+1] $. Then,
since
\[ 
\sum_{l=1}^{\noTot+1}  N[i,l] Q[l,j]
=
\sum_{l=1}^{\noTot+1} \tilde N[i,l] d''(l) Q[l,j]
=
\sum_{l=1}^{\noTot+1}   \tilde N[i,l]   Q[j,l]  d''(j),
\]
we can rewrite
\eqref{eq:qn=nq} as
\begin{equation}\label{eq:qtn=tnq}
\tilde N[i,j] =
\sum_{l=1}^{\noTot+1} \tilde N[i,l] Q[j,l] +\frac{ \delta_{ij}}{d''(j)}
=
\sum_{l=1}^{\noTot+1} Q[i,l] \tilde N[l,j] + \frac{\delta_{ij}}{d''(j)}
\quad \text{ for all } i,j \in [\noTot+1].
\end{equation}
It follows that for all $ i \in [\noTot+1] $, the function  $ f_i  : [\noTot+1] \to \mbb R $   defined by
\[ 
f_i (l)  = \tilde N[i,l] = \tilde{N}[l,i], \quad l \in [\noTot+1],
\]
is harmonic on the set $  [\noTot] \setminus \lrb{i} $ (the function is well defined, since $ \tilde N[i,l] = \tilde{N}[l,i] $ due to the symmetry of $ \tilde N $). 

In particular, $ f_{\noTot+1} $  is harmonic on the set $ [\noTot] $, whose boundary is the singleton $ \lrb{\noTot+1} $. Hence $ f_{\noTot+1} $ is a constant function.
Similarly, $ f_1 $   is constant on $ [\noTot] $, because it is harmonic on  $  [2 \tdots \noTot] $, whose boundary is the singleton $ \lrb{1} $.

\begin{corollary}\label{th:cor4}
	For all $ i \in [\noTot+1] $ we have
	\[ 
	\tilde N[i,\noTot+1]  =    N[\noTot+1,i] = \frac{1}{d_1}
	\]
	and for all $ i \in [\noTot] $ we have
	\[ 
	\tilde N[i,1]  =    N[1,i] = \alpha^2+\frac{1}{d_1}.
	\]
\end{corollary}
\begin{proof}
	We already concluded that $ f_{\noTot+1} $ is a constant function, i.e., the value
	\[ 
	f_{\noTot+1}(j) = \tilde N[\noTot+1,j] = \tilde N[j,\noTot+1] 
	\]
	does not depend on $ j $. By Lemma \ref{th:L1},
	\[ 
	\tilde N[\noTot+1,\noTot+1] = \frac{1}{d''(\noTot+1)} N[\noTot+1,\noTot+1] = \frac{1}{d_1  + \alpha^{-2}}  \cdot  \frac{d_1 \alpha^2+1}{d_1 \alpha^2} = \frac{1}{d_1},
	\]
	and the first  claim follows.
	
	The other claim follows from the fact that $ f_1 $   is constant on $ [\noTot] $ and 
	\[ 
	f_1(1) =  \frac{1}{d_1  + \alpha^{-2}}  N[1,1]=  \frac{1}{d_1  + \alpha^{-2}}   \cdot    \frac{ \lr{d_1 \alpha^2+1}^2}{d_1 \alpha^2} = \alpha^2 + \frac{1}{d_1}.
	\]
\end{proof}

It remains to describe the values of $ \tilde N[i,j] $ when $ 2 \leq i, j \leq \noTot$. For each $ i\in [2 \tdots \noTot] $ define  $ \phi_i $ to be the unique function which
\begin{itemize}
	\item is harmonic on $ [2 \tdots \noTot] \setminus \lrb i $,
	\item satisfies $ \phi(1) = 1 $, $ \phi(i) =0$.
\end{itemize}
Let $ R : [2 \tdots \noTot]  \to \mbb R$ be defined by
\[ 
R(i) =      \lr{\sum_{j: v_j \sim v_i}  \phi_i(j)   }^{-1}.
\]

\begin{lemma}\label{th:tN_general}
	For all $ i \in [2 \tdots \noTot] $, $ j \in [\noTot] $, it holds that
	\[ 
	\tilde N[i,j]  = 
	\alpha^2  + \frac{1}{d_1} + \lr{1 - \phi_i(j)} R(i).
	\]
\end{lemma}
\begin{proof}
	Fix $ i \in [2 \tdots \noTot] $. Let $ m  =f_i(1) $ (we already have $ m = \alpha^2  + {d_1}^{-1} $) and $ M=f_i(i)-m $ ($ M $ to be described). Then $ f_i $ is the unique function which is harmonic on $ [2 \tdots \noTot] \setminus \lrb i$ and satisfies the boundary conditions $ f_i(1) = m $, $ f_i(i) = m+M $. Clearly, $ 0 \neq M $, since otherwise $ f_i $ must be a constant (and therefore harmonic) function, but from \eqref{eq:qtn=tnq} it follows that $ f_i $ is not harmonic at $ i $. 
	
	Define 
	\[ 
	g(j) = \frac{1}{M} \lr{ f_i(j)  -m}, \quad  j \in [\noTot].
	\]
	Then $ g $ is harmonic on $[2 \tdots \noTot] \setminus \lrb i $ and satisfies the boundary conditions $ g(1) = 0 $, $ g(i) = 1$. By the Uniqueness Principle, $ g \equiv 1-\phi_i  $, since $ 1-\phi_i $ satisfies the same conditions. Hence
	\[ 
	f_i(j)   = m + M \lr{1 - \phi_i(j)} , \quad j \in [\noTot].
	\]
	From \eqref{eq:qtn=tnq} we have
	\[ 
	f_i(i) =
	\frac{1}{d''(i)} \lr{1 + \sum_{j:v_j \sim v_i}  f_i(j)}
	= 
	m +  M+ \frac{1}{d''(i)} \lr{1 - M  \sum_{j:v_j \sim v_i}  \phi_i(j)}.
	\]
	On the other hand, $ f_i(i)  = m+M$. Taking into account the definition of $ R(i) $, we have
	\[ 
	m+M = m +M + \frac{1}{d''(i)} \lr{1 - \frac{M}{R(i)} }
	\]
	or $ M = R(i) $, which concludes the proof.
\end{proof}

\begin{lemma}\label{th:tree}
	Suppose that the original graph $ \mc G $ is a tree. 
	For all vertices $ v_i, v_j \in \vertSet$, let   $ \ell(i,j) $ be the distance from the lowest common ancestor of $ v_i, v_j $  to the root $ v_1 $. 
	
	Then for all $ i,j \in [\noTot] $ it holds that
	\[ 
	\tilde N[i,j]  =  \alpha^2 +  \frac{1}{d_1} +  \ell(i,j).
	\]
\end{lemma}
\begin{proof}

	For $ i  =1 $ the claim follows from Corollary \ref{th:cor4}. Let $ i  \in [2 \tdots \noTot]$, then from Lemma \ref{th:tN_general} we have
	\[ 
	f_i(j) =   \alpha^2 +  \frac{1}{d_1} +  (1-\phi_i(j)) R(i), \quad j \in [\noTot].
	\]
	Let us show that 
	\begin{equation}\label{eq:tree1}
	\phi_i(j) = 1 - \frac{\ell(i,j)}{\ell(i)}, \quad j \in [\noTot].
	\end{equation}
	and
	\begin{equation}\label{eq:tree2}
	R(i)  = \ell(i).
	\end{equation}
	
	In \eqref{eq:tree1}  the boundary conditions $ \phi_i(1) = 1 $, $ \phi_i(i)=0 $ are satisfied. By the Uniqueness Principle it remains to show that the right-hand side of \eqref{eq:tree1} defines a harmonic function in $ j $ on $ [2 \tdots \noTot] \setminus \lrb i$. Equivalently, we must show that $ \ell(i, \cdot)  $ is harmonic on $ [2 \tdots \noTot] \setminus \lrb i$.
	
	Fix any $ j \in [2 \tdots \noTot] \setminus \lrb i $.
	There are two cases to consider:
	\begin{itemize}
		\item The vertex $ v_j $ is not on  the path from the root to $ v_i $; then the lowest common ancestor of $v_ i $ and $ v_j $ coincides with the lowest common ancestor of $ v_i $ and the parent of $ v_j $ or the lowest common ancestor of $v_ i $ and any child of $ v_j $; hence $ \ell(i,j) = \ell(i,k) $ for all vertices $ v_k $, adjacent to $ v_j $. 
		\item The vertex $ v_j $ is   on  the path from the root to $ v_i $. Let $ v_p $ be the parent of $v_ j $ and $ v_{c}  $ be the unique child of $ v_j $ which also is on  the path from the root to $ v_i $.  Then for each vertex $ v_k \sim v_j $ we have
		\[ 
		\ell(i,k)  =
		\begin{cases}
		\ell(i,j) , & k  \notin \lrb{p,c}, \\
		\ell(i,j)+1 , & k =c, \\
		\ell(i,j) -1, & k  =p.
		\end{cases}
		\]
	\end{itemize}
	In both cases
	we obtain
	\[ 
	\sum_{k :  v_k \sim v_j}\frac{\ell(i,k)}{d''(j)}  = \ell(i,j),
	\]
	i.e., $ \ell(i,\cdot) $ (and  thus also $ 1 - \frac{\ell(i,\cdot)}{\ell(i)} $) is harmonic at every  $ j  \in[2 \tdots \noTot] \setminus \lrb i$. We conclude that \eqref{eq:tree1} holds.
	
	It remains to show  \eqref{eq:tree2}. From the definition of $ R $,
	\[ 
	R(i)^{-1}= \sum_{j: v_j \sim v_i}  \phi_i(j)   = d''(i) -   \frac{1}{\ell(i)}\sum_{j: v_j \sim v_i}  \ell(i,j).
	\]
	On the other hand, for every vertex $v_ j \sim v_i $ we have
	\[ 
	\ell(i,j)  =
	\begin{cases}
	\ell(i) -1, & v_j  \text{ is  the parent of } v_i,\\
	\ell(i) , &  v_j \text{ is a child of } v_i.
	\end{cases}
	\]
	Thus 
	\[ 
	\sum_{j: v_j \sim v_i}  \ell(i,j) =  d''(i)  \ell(i) - 1
	\quad \text{and} \quad
	R(i)^{-1}=\ell(i)^{-1}.  
	 \]
	Now, by combining \eqref{eq:tree1} and \eqref{eq:tree2} with Lemma \ref{th:tN_general}, we obtain the desired equality.

	\textbf{Remark.} For another argument why $ R(i)  = \ell(i)$, see   \cite[Exercise 9.7]{Levin2009}.
\end{proof}

\begin{corollary}\label{th:DAG}
	Suppose that $ \mc G $ is an arbitrary DAG satisfying the initial assumptions.

	Then  for all $ i,j \in [\noTot] $ we have
	\[ 
	\tilde N[i,j]  - \lr{\alpha^2 +  \frac{1}{d_1} } \leq   \ell(i).
	\]
	In particular,  
	\[ 
	0 \leq  \tilde N[i,j]  - \lr{\alpha^2 +  \frac{1}{d_1} } \leq   n
	\]
	for all $ i,j \in [\noTot] $. Moreover, when $ i=1 $ or $ j=1 $, the  equality  $   \tilde N[i,j]  =  {\alpha^2 +  \frac{1}{d_1} }  $ holds.
\end{corollary}
\begin{proof}
	When $ i=1 $ or $ j=1 $, the assertion holds (Lemma \ref{th:cor4}). Suppose that $ i>1 $. Since $ \tilde N[i,\cdot] $ is harmonic on $[2 \tdots \noTot] \setminus \lrb i $, it attains its extreme values on the boundary $ \lrb{1,i} $, i.e., it suffices to show the inequality for $ j=i $. In view of Lemma \ref{th:tN_general}, this becomes $ R(i) \leq \ell(i) $.
	
	Let $ \mc T $ be any spanning tree of $ \mc  G $ s.t. the shortest path between $ v_i $ and the root is preserved, i.e., distance from $ v_i $ to $ v_1 $ is still $ \ell(i) $.  By replacing $ \mc G $ with $ \mc T $, the value of $ R(i) $  can only increase, since replacing $ \mc G $ with $ \mc T$ corresponds to deleting edges in $ \mc G $ and this operation, by Rayleigh's Monotonicity Law, can only increase the effective resistance $ R(i) $. However, Lemma \ref{th:tree} ensures that for the tree $ \mc T$ the value of $ R(i) $  equals $ \ell(i) $. Thus in the graph $\mc G $ the value   $ R(i) $ is upper-bounded by $ \ell(i) \leq n $.
\end{proof}

\subsection{Entries of the matrix \texorpdfstring{$ K $}{K} }\label{sec:appE}
\label{sec:entries}

Now we shall describe the matrix $ K $, where $ K $, $ L $, $ \mata $, $ \matb $, $ a $ and $ b $ are defined as in Section \ref{sec:Ra}.
The adjacency matrix $ \Gamma $ of the graph  $ \mc G' $ has the following block structure:
\[ 
\Gamma 
=
\begin{pmatrix}
\mb 0_{A, A} & H& \alpha^{-2} \mb e \\[0.3cm]
H^*    \\ 
& &  \mb 0_{B+1, B+1}   \\
\alpha^{-2}  \mb e^* 
\end{pmatrix},
\]
where $ H $ is a matrix of size $ A \times B $ and
$ \mb e $ stands for the column vector of length $ A $, whose first entry is 1 and the remaining entries are 0.

We claim that 
\begin{lemma}
	\[ 
	H=\mata^* \matb.
	\]
\end{lemma}
\begin{proof}
	We  have to show that for every $ i \in [A] $, $ j\in [A+1 \tdots \noTot] $ it holds that
	\begin{equation}\label{eq:l17-1}
	\Gamma[i,j] = \sum_{e \in [T+1]} \mata[e,i] \matb[e,j].
	\end{equation}
	Notice that $ \Gamma[i,j] = 0 $ unless $ v_i \sim v_j $   (in that case $ \Gamma[i,j] = 1 $ for $ i \in [A] $, $ j\in [A+1 \tdots \noTot] $). On the other hand,  both $ \mata[e,i] $  and $ \matb[e,j] $ are simultaneously nonzero iff the edge  $ e $ is  incident both to   $ v_i $  and $ v_j $.
	
	There are two cases to consider:
	\begin{enumerate}
		\item $ v_i \sim v_j $; then $ \Gamma[i,j] = 1$ and there is a unique edge $ e \in [T+1] $ s.t. $ \mata[e,i] \neq 0 $ and $ \matb[e,j]  \neq 0 $. Since $ i \in [A] $, $ j\in [A+1 \tdots \noTot] $, we have $\mata[e,i] =  \matb[e,j] =1 $ and \eqref{eq:l17-1} holds.
		\item $ v_i \not \sim v_j $; then  $ \Gamma[i,j] = 0 $ and for each edge $ e \in [T+1] $ either $ \mata[e,i] = 0 $ or $ \matb[e,j]  =0 $. Again, \eqref{eq:l17-1} holds.
	\end{enumerate}
	
\end{proof}

Denote
\[ 
\matD= \diag(\vecP)^{-1} \Gamma  \diag(\vecP)^{-1}  =  \diag(\vecP)  Q  \diag(\vecP)  ,
\]
then $ \matD $ has the following block structure:
\[ 
\matD 
=
\begin{pmatrix}
\mb 0_{A, A} &\tilde L \\ 
\tilde L ^*    &  \mb 0_{B+1, B+1}   
\end{pmatrix},
\]
where   $ \tilde L $ is an $ A \times (B+1)  $ matrix with the following block structure:
$ \tilde L = 
\begin{pmatrix}
L & \beta \mb e
\end{pmatrix} $. This follows from the fact that  $ \vecP $ has the following block structure:
\[ 
\vecP  = \begin{pmatrix}
a \\ b \\ a[1]
\end{pmatrix},
 \]
and from the  equalities
\[ 
 L = \diag(a)^{-1} \mata^* \matb \diag(b)^{-1}  =  \diag(a)^{-1} H  \diag(b)^{-1}   
 \quad\text{and}\quad
  \beta = \lr{\alpha  a[1]}^{-2}.
 \]

Now we can show the following characterization of the matrix $ K = \lr{ \mb I_{A}    -  L   L^* }^{-1} $:
\begin{lemma}\label{th:K_estm}
	Let $ \tNAA $ stand for the leading $  A\times A $ principal submatrix of $ \tilde N $, i.e.,  for  the submatrix of $ \tilde N $, formed by the rows indexed by $  [A]$ and columns indexed by $ [A] $.
	Then
	\[ 
	K = 
	\diag(a)   \lr{ \tNAA   - d_1^{-1} \mb J}  \diag(a) ,
	\]
	where $ \mb J $ is the  $ A \times A $ all-ones matrix.
\end{lemma}
\begin{proof}
	Since $  Q = \diag(\vecP)^{-1} \matD  \diag(\vecP) $ and $ N  = \lr{\mb I_{\noTot+1} - Q}^{-1} $, we have
	\begin{equation}\label{eq:l9e1}
	\lr{\mb I_{\noTot+1} - \matD}^{-1} 
	=
	\diag(\vecP) N \diag(\vecP)^{-1}   
	=
	\diag(\vecP) \tilde  N \diag(\vecP) .
	\end{equation}

	By the block-wise inversion formulas \cite[Proposition 2.8.7]{bernstein2009matrix}, $ \lr{ \mb I_{\noTot+1} - \matD}^{-1} $ has the following block structure:
	\[ 
	\lr{ \mb I_{\noTot+1} - \matD}^{-1}
	=
	\begin{pmatrix}
	\lr{\mb I_{A}    - \tilde L \tilde L^*}^{-1}  &  -\lr{\mb I_{A}    - \tilde L \tilde L^*}^{-1} \tilde L  \\
	- \tilde L^*\lr{\mb I_{A}    - \tilde L \tilde L^*}^{-1} & \lr{\mb I_{B+1}    - \tilde L^* \tilde L}^{-1}
	\end{pmatrix}.
	\]
	From \eqref{eq:l9e1} we  conclude that
	\begin{equation}\label{eq:l9e2}
	\lr{\mb I_{A}    - \tilde L \tilde L^*}^{-1} 
	=
	\diag(a)  \tNAA \diag(a)  .
	\end{equation}

	From the Sherman-Morrison formula \cite[Fact 2.16.3]{bernstein2009matrix}  we have that  for an invertible matrix $ W $ and a column vector $ w $ s.t. $ 1+ w^* W^{-1}w \neq 0  $ the inverse of the updated matrix  $ W  + ww^* $ can be computed as
	\begin{equation*}\label{eq:l9e3}
	\lr{ W  + ww^*  }^{-1}   =  W^{-1}  - \frac{W^{-1} ww^* W^{-1}  }{1+ w^* W^{-1}w}.
	\end{equation*}
	Take $ W =  \mb I_{A}    - \tilde L \tilde L^*$ and $ w = \beta \mb e $; then
	\begin{itemize}
		\item $ \tilde L \tilde L^* =  L  L^* + \beta^2 \mb e \mb e^* $;
		\item  $ \mb e \mb e^* $ is a matrix of size $ A \times A $, whose only nonzero entry is 1 in the first row and column;
		\item $ w^* W^{-1}w = \beta^2 W^{-1}[1,1]  = \beta^2 d''(1)  \tilde N[1,1] =   \frac{1}{d_1 \alpha^2} $ (by \eqref{eq:l9e2});
		\item $ 	1+w^* W^{-1}w  = \frac{1}{1 - \beta} \neq 0 $;
		\item $  \mb I_{A}    -  L   L^*  = W + ww^*  $.
	\end{itemize}
	We obtain
	\[ 
	K = \lr{ \mb I_{A}    -  L   L^* }^{-1} =     \lr{\mb I_{A}    - \tilde L \tilde L^*}^{-1} \lr{ \mb I_{A}  - \beta^2(1-\beta) \mb e \mb e^*  W^{-1}  }.
	\]
	Applying \eqref{eq:l9e2}  
	yields
	\[ 
	K = 
	\diag(a)  \tNAA   U  \diag(a)  ,
	\]
	where
	\begin{align*}
	U & :=
	\diag(a)  \lr{ \mb I_{A}  -  \beta^2(1-\beta) \mb e \mb e^*  W^{-1}  }\diag(a)^{-1}    \\
	& =\mb I_{A}  -\beta^2(1-\beta) \mb e \mb e^*    \diag(a)  W^{-1} \diag(a)^{-1}   \\
	& =\mb I_{A}  -\beta^2(1-\beta) \mb e \mb e^*    \diag(a)  \tNAA  \\
	& =\mb I_{A}  -\beta^2(1-\beta)  d''(1)  \mb e \mb e^*     \tNAA .
	\end{align*}
	Here we have used the fact that $  \mb e \mb e^*  $ and $   \diag(a)   $ commute and   $ \mb e \mb e^*    \diag(a)^2  = d''(1)  \mb e \mb e^* $.
	It follows that
	\[ 
	\tNAA   U
	=
	\tNAA  -\beta^2(1-\beta)  d''(1)    \tNAA  \mb e \mb e^*     \tNAA   .
	\]

	Furthermore, the row vector $  \mb e^*     \tNAA  $  equals the first row of $ \tNAA $ (and $ \tNAA $ is a symmetric matrix).
	All elements of the first row of $ \tNAA $  are equal to $ \alpha^2 + d_1^{-1} $, hence
	\[ 
	\lr{  \tNAA  \mb e} \lr{\mb e^*     \tNAA} = \lr{\alpha^2 + d_1^{-1}}^2  \mb J.
	\]
	It is straightforward to check that
	\begin{align*}
	& \beta (1-\beta)d''(1) \lr{\alpha^2 + d_1^{-1}}  
	=1,\\
	&  \beta^2 (1-\beta)d''(1) \lr{\alpha^2 + d_1^{-1}} ^2
	=
	\beta  \lr{\alpha^2 + d_1^{-1}}   = d_1^{-1},
	\end{align*}
	thus
	\[ 
	\tNAA   U
	=
	\tNAA   - d_1^{-1} \mb J 
	\quad
	\text{and}
	\quad
	K = 
	\diag(a)   \lr{ \tNAA   - d_1^{-1} \mb J}  \diag(a) .
	\]
\end{proof}

We now continue with the proof of Lemma \ref{lem:corr9}, restated here for   convenience.
\begingroup
\def\thelemmaPrime{\ref*{lem:corr9}}
\begin{lemmaPrime}
		For all $ i,j  \in  [A]$    the following inequalities hold:
		\[ 
		\alpha^2     a[i] a[j]
		\leq
		K[i,j] \leq  
		\lr{\alpha^2    + n} a[i] a[j]
		.
		\]
		Moreover, when $ i=1 $ or $ j=1 $, we have $ K[i,j] = {\alpha^2    } a[i] a[j] $.
\end{lemmaPrime}
\endgroup
\begin{proof}
	From Lemma \ref{th:K_estm}   we have
	\[ 
	K[i,j] = \lr{ \tilde N[i,j]  - \frac{1}{d_1}  } a[i] a[j]  , \quad i ,j \in [A].
	 \]
	On the other hand,  from Corollary \ref{th:DAG} we have 
	\[ 
	\alpha^2  \leq  \tilde N[i,j]  -   \frac{1}{d_1}   \leq   \alpha^2 +n , \quad i ,j \in [A],
	\]
	and $   \tilde N[i,j]  =  {\alpha^2 +  \frac{1}{d_1} }  $ whenever $ i =1$ or $ j=1 $. This proves the claim.
\end{proof}

\section{Proofs for AND-OR tree evaluation}
\label{app:andor}

\begin{proof} 
[Proof of Lemma \ref{lem:heavy}]

\noindent {\bf Correctness.}
Children of a vertex $v$ gets added to $\T'$ if the tree size estimate for $\T(v)$  is at least $\frac{2m}{3}$. If the tree size estimation in step 1 is correct, we have the following:
\begin{itemize}
\item
if the actual size is at least $m$, the estimate is at least $(1-\delta)m = \frac{2m}{3}$;
\item 
if the actual size is less than $\frac{m}{2}$, the estimate is less than 
$(1+\delta)\frac{m}{2} = \frac{2m}{3}$.
\end{itemize} 
This means, if all the estimates are correct, then $\T'$ is a correct $m$-heavy element subtree.

To show that all the estimates are correct with probability at least $1-\epsilon$, we have
to bound the number of calls to tree size estimation. We show

\begin{lemma}
An $m$-heavy element subtree $\T'$ contains at most $n \frac{6T}{m}$ vertices.
\end{lemma}

\begin{proof}
On each level, $\T'$ has at most $\frac{2T}{m}$ vertices $x$ with $|\T(x)|\geq \frac{m}{2}$.
Since the depth of $\T$ (and, hence, $\T'$) is at most $n$, the total number of such vertices is at most
$\frac{2Tn}{m}$. For each such $x$, $\T'$ may also contain its children with $|\T(x)|< \frac{m}{2}$.
Since each non-leaf $x$ has two children, the number of such vertices is at most $\frac{4Tn}{m}$
and the total number of vertices is at most $\frac{6Tn}{m}$.
\end{proof}

Since the algorithm makes one call to tree size estimation for each $x$ in $\T'$ and
each call to tree size estimation is correct with probability at least  
$1-\epsilon'=1- \frac{m}{6nT} \epsilon$,
we have that all the calls are correct with probability at least $1-\epsilon$.

\noindent {\bf Running time.}
Since the formula tree is binary, we have $d=O(1)$. 
Also, $\delta = \frac{1}{3}$ is a constant, as well.
Therefore, each call to tree size estimation uses
\[ O\left( \sqrt{nm} \log^2 \frac{1}{\epsilon'} \right) \]
queries. Since tree size estimation is performed only for vertices that get added to $\T'$, it is performed at most $ \frac{6Tn}{m}$ times. Multiplying the complexity of one call by $ \frac{6Tn}{m}$  and
substituting $\epsilon' = \frac{m}{6nT} \epsilon > \frac{\epsilon}{6 T^2}$ gives the complexity of
\[ O\left( \frac{n^{1.5} T}{\sqrt{m}}  \log^2 \frac{T}{\epsilon} \right) .\]
\end{proof}

{\bf Analysis of the main algorithm: correctness.}

\begin{lemma}
If {\bf Unknown-evaluate} is used as a query, it performs a transformation $\ket{\psi_{start}} \rightarrow \ket{\psi}$ where $\ket{\psi}$ satisfies 
\[ \| \psi - (-1)^T \psi_{start} \| \leq 2\sqrt{\epsilon} \]
\end{lemma}

\begin{proof}
We use induction over $i\in\{0, 1, \ldots, c\}$, with $i=0$ being queries to one variable (which are invoked by 
{\bf Unknown-evaluate}($r, 1, \epsilon$) in the 3$ ^{\text{rd}} $ step when $i=1$) and $i=1, \ldots, c$ being calls to {\bf Unknown-evaluate} 
with the respective value of $i$.

For the base case ($i=0$), a query to a variable produces the transformation $\ket{i}\rightarrow (-1)^{x_i}\ket{i}$.
In this case, $T=x_i$, so, this is exactly the correct transformation.

For the inductive step ($i\geq 1$),  we first assume that, 
instead of calls to {\bf Unknown-evaluate} at the leaves, we have perfect queries with no error.
Let $\ket{\psi_{ideal}}$ be the final state under this assumption. We first bound
$\| \psi_{ideal} - (-1)^T \psi_{start} \|$ and then bound the difference between $\ket{\psi_{ideal}}$ and
the actual final state $\ket{\psi}$.

Let $\ket{\psi'}=\sum_{\T', x} \alpha_{\T', x} \ket{\T', x} \ket{\psi_{\T', x}}$ be the state after the first three steps, with $\T'$ being the subtree obtained in the 1$ ^\text{st} $ or the 2$ ^\text{nd} $ step, $x$ being the result obtained at the 3$ ^\text{rd} $ step and 
$\ket{\psi_{\T', x}}$ being all the other registers containing intermediate information. We express $\ket{\psi'}=\ket{\psi_1}+\ket{\psi_2}+\ket{\psi_3}$ 
where
\begin{itemize}
\item
$\ket{\psi_1}$ contains terms where $\T'$ is a valid $m$-heavy element subtree and $x$ is the correct answer;
\item
$\ket{\psi_2}$  contains terms where $\T'$ is a valid $m$-heavy element subtree but $x$ is not the correct answer;
\item
$\ket{\psi_3}$ contains all the other terms.
\end{itemize}
By the correctness guarantees for steps 2 and 3, we have $\|\psi_2\|\leq \sqrt{\epsilon/5}$ and 
$\|\psi_3\|\leq \sqrt{\epsilon/5}$.

After the phase flip in step 4, the state becomes $\ket{\psi''} = (-1)^T \ket{\psi_1} + \ket{\psi'_2} + \ket{\psi'_3}$ with
$\ket{\psi'_2}$ and $\ket{\psi'_3}$  consisting of terms from $\ket{\psi_2}$ and $\ket{\psi_3}$ , with some of their phases 
flipped. Hence, 
\[ \| \psi''- (-1)^T \psi' \|  \leq \| \psi'_2 + \psi'_3 \| + \| \psi_2 + \psi_3 \|  
 \leq \|\psi'_2\|+\|\psi'_3\|+\|\psi_2\|+\|\psi_3\|  \leq 4 \sqrt{\frac{\epsilon}{5}} . \]
Reversing the first three steps maps $\ket{\psi'}$ back to $\ket{\psi_{start}}$ and $\ket{\psi''}$ to $\ket{\psi_{ideal}}$.
Hence, we have the same estimate for $\| \psi_{ideal} - (-1)^T \psi_{start}\|$.
  
We now replace queries by applications of  {\bf Unknown-evaluate}. Let $t=O(\sqrt{sn})=O(s)$ be the number of queries. 
Let $\ket{\phi_i}$ be the final state of the algorithm if the first $i$ queries use the perfect query transformation and the remaining queries are implemented using {\bf Unknown-evaluate}. Then, $\ket{\psi_{ideal}}=\ket{\phi_t}$ and $\ket{\psi}=\ket{\phi_0}$ and
we have
\[ \| \psi - \psi_{ideal} \| = \| \phi_0 - \phi_t \| 
\leq \sum_{j=0}^{t-1} \| \phi_j - \phi_{j+1} \|     
 \leq \frac{2s\sqrt{\epsilon}}{s^{3/2}} = o(\sqrt{\epsilon}) ,\]
with the second inequality following from the inductive assumption. 
(The only difference between $\ket{\phi_j}$ and $\ket{\phi_{j+1}}$ is that, in the first case, we apply a perfect query transformation in the $(j+1)^{\rm st}$ query and, in the second case, we apply {\bf Unknown-evaluate}($v, {i-1}, \epsilon/s^3$) instead. The distance between the states resulting from these two transformations can be bounded by 
$\frac{2\sqrt{\epsilon}}{s^{3/2}}$ by the inductive assumption.)
Therefore,
\[ \| \psi - (-1)^T \psi_{start} \| \leq \| \psi_{ideal} - (-1)^T \psi_{start} \| + \| \psi-\psi_{ideal} \| 
\leq \frac{4}{\sqrt{5}} \sqrt{\epsilon} + o(\sqrt{\epsilon}) < 2 \sqrt{\epsilon}. \]
This concludes the proof.
\end{proof}

For the case when {\bf Unknown-evaluate} is used to obtain the final answer,
let $\ket{\psi'}$ be the final state if, instead of calls to {\bf Unknown-evaluate} at the next level,
we had perfect queries and let $\ket{\psi}$ be the actual final state of the algorithm. 
We express $\ket{\psi}=\ket{\psi_{cor}}+\ket{\psi_{inc}}$ where $\ket{\psi_{cor}}$ ($\ket{\psi_{inc}}$) is the part of
the state where the algorithm outputs correct (incorrect) answer.
Let $\ket{\psi'}=\ket{\psi'_{cor}}+\ket{\psi'_{inc}}$ be a similar decomposition for $\ket{\psi'}$.
Then, $\|\psi'_{inc}\|\leq \|\psi_2\|+\|\psi_3\| \leq \frac{2}{\sqrt{5}} \sqrt{\epsilon}$ and
\[ \|\psi_{inc} \| \leq \|\psi'_{inc}\| + \| \psi_{inc} - \psi'_{inc} \| 
\leq \|\psi'_{inc}\| + \| \psi - \psi' \| \leq  \frac{2}{\sqrt{5}} \sqrt{\epsilon} + o(\sqrt{\epsilon}) < \sqrt{\epsilon} .\]
The probability of {\bf Unknown-evaluate} outputting an incorrect answer is $\|\psi_{inc} \|^2 < \epsilon$.

{\bf Analysis of the main algorithm: running time.}
\begin{lemma}
\label{lem:26}
The number of queries made by {\bf Unknown-evaluate}($r, i, \epsilon$) is of an order
\[ O\left( n^i \sqrt{T_i T^{1/c}} \left(\log T_i + \log \frac{1}{\epsilon}\right)^i \right) .\]
\end{lemma}

\begin{proof}
Generating $\T'$ takes $O(T_1)$ steps if $i=1$ and $O(\frac{n^{1.5} T_i}{\sqrt{T_{i-1}}} \log^2 \frac{T_i}{\epsilon})$ 
steps if $i>1$.
Since $\frac{T_i}{\sqrt{T_{i-1}}} = \sqrt{T_i T^{1/c}}$, this is at most the bound in the statement of the lemma.

In the next step, the algorithm calls {\bf Unknown-evaluate} for $O(\sqrt{n s} \log \frac{1}{\epsilon})$ 
subtrees where $s$ is the size of $\T'$.
Since $s=O\left(n \frac{T_i}{T_{i-1}}\right)$, this means $O\left(n \frac{\sqrt{T_i}}{\sqrt{T_{i-1}}}\log \frac{1}{\epsilon}\right)$
calls of {\bf Unknown-evaluate}. For each of them, the number of queries is 
\[ O\left( n^{i-1} \sqrt{T_{i-1} T^{1/c}} \left(\log T_{i-1} + \log \frac{1}{\epsilon'}\right)^{i-1} \right) =
O\left( n^{i-1} \sqrt{T_{i-1} T^{1/c}} \left(\log T_{i} + \log \frac{1}{\epsilon}\right)^{i-1} \right) .\]
Multiplying the number of calls to {\bf Unknown-evaluate} with the complexity of each call gives the claim.
\end{proof}

We now show how Lemma \ref{lem:26} implies Theorem \ref{th:andor}.
If $i=c$, the expression of Lemma \ref{lem:26} is equal to
\[ O\left( n^{c} \sqrt{T^{1+1/c}} \left(\log T + \log \frac{1}{\epsilon}\right)^{c} \right) =
O\left( T^{\frac{1}{2}+\frac{1}{2c}+o(1)} \right) .\]
Since $c$ can be chosen arbitrarily large,  we can achieve $O(T^{1/2+\delta})$ for arbitrarily small $\delta>0$.
The total running time is $O(\log T)=O(T^{o(1)})$ times the number of queries.

\section{Conclusion}

Search trees of unknown structure (which can be discovered by local exploration) are commonly used in classical computer science.
In this paper, we constructed three quantum algorithms for such trees: for estimating size of a tree, for finding whether a tree contains a marked vertex (improving over an earlier algorithm by Montanaro) and for evaluating an AND-OR formula described by a tree of unknown structure.

Some of possible directions for future study are:
\begin{enumerate}
\item
{\bf Space-efficient algorithm for AND-OR formula evaluation?} 
Our algorithm for evaluating AND-OR formulas in Section \ref{sec:andor} uses a substantial amount of memory to store the heavy element subtrees. 
In contrast, the algorithm of \cite{A+} for evaluating formulas of known structure only uses $O(\log T)$ qubits of memory.  
Can one construct an algorithm for evaluating formulas of unknown structure with time complexity similar to our algorithm but smaller space requirements?
\item
{\bf Speeding up other methods for solving NP-com\-ple\-te problems.}
Backtracking is used by SAT solvers and other algorithms for NP-complete problems. 
Our quantum algorithm for backtracking provides an almost quadratic quantum improvement over those algorithms. What other methods for solving NP-com\-ple\-te problems have faster quantum counterparts?
\item
{\bf Other parameter estimation problems.}
The tree size estimation problem can be viewed as a counterpart of quantum counting, in a more difficult setting. 
What other problems about estimating size (or other parameters) of combinatorial structures would be interesting and what would they be useful for?
\end{enumerate}

\section*{Acknowledgements}
{The authors would like to thank Mark Goh for his valuable comments on a previous version of the paper, in particular, for pointing out a mistake in the proof of Lemma \ref{th:lemma14}. The corrected version requires adjusting the bound on $\lrv{\hat \theta - \theta }$ in Lemma \ref{lem:theta} and choosing the parameter value $\delta_{min}=\frac{\delta^{1.5}}{ 24\sqrt { 3nT_0 }}$ (instead of $\delta_{min}=\frac{\delta^{1.5}}{ 4\sqrt { 3nT_0 }}$) in Algorithm \ref{alg:dag}.}

{This work has been supported by the ERC Advanced Grant MQC and Latvian State Research programme NexIT project No.1.}

\end{document}